\documentclass{LMCS}

\def\dOi{10(2:10)2014}
\lmcsheading%
{\dOi}
{1--28}
{}
{}
{Jul.~\phantom01, 2013}
{Jun.~16, 2014}
{}

\usepackage{hyperref,algorithm,algorithmic,amsmath,mathabx,stmaryrd,fancybox,tikz,wrapfig,amsthm,float}

\theoremstyle{theorem}

\theoremstyle{definition}

\newtheorem{lemma}{Lemma}

{\unskip\nobreak\hskip 1em plus 1fil\nobreak$\Box$
\parfillskip=0pt%
\endtrivlist}

\newcommand{\cpath}{\alpha}
\newcommand{\p}{p}
\newcommand{\x}{x}
\newcommand{\f}{\phi}
\newcommand{\ff}{\psi}
\newcommand{\fff}{\varphi}
\newcommand{\mf}[2]{\langle #1\rangle #2}
\newcommand{\muf}[2]{\mu #1.#2}
\newcommand{\cf}[2]{{#1\! #2}}
\newcommand{\m}{m}
\newcommand{\maxk}{{maxK}}
\newcommand{\subst}[2]{\left[^{#1}/_{#2}\right]}
\newcommand{\flag}[2]{#1^{#2}}
\newcommand{\flc}{\text{FL}}
\newcommand{\lean}{lean}
\newcommand{\fc}{\shortdownarrow}
\newcommand{\fs}{\shortrightarrow}
\newcommand{\ifc}{\shortuparrow}
\newcommand{\ifs}{\shortleftarrow}

\newcommand{\nod}{n}
\newcommand{\Nd}{N}
\newcommand{\dual}[1]{\overline{#1}}
\newcommand{\Nodes}{\mathcal{N}}
\newcommand{\Rel}{\mathcal{R}}
\newcommand{\Lab}{\mathcal{L}}
\newcommand{\Prop}{\mathcal{P}}
\newcommand{\M}{M}
\newcommand{\sem}[3]{[\![#1]\!]^{#2}_{#3}}
\newcommand{\semm}[3]{[\![\![#1]\!]\!]^{#2}_{#3}}
\newcommand{\nom}{o}
\newcommand{\R}{R^{FL}}

\begin{document}
\title[Global Numerical Constraints on Trees]{Global Numerical
  Constraints on Trees} 
\author[E.~B\'arcenas]{Everardo B\'arcenas\rsuper a}
\address{{\lsuper a}Universidad Polit\'ecnica de Puebla, M\'exico}
\email{ismael.barcenas@uppuebla.edu.mx}

\author[J.~Lavalle]{Jes\'us Lavalle\rsuper b}
\address{{\lsuper b}Benem\'erita Universidad Aut\'onoma de Puebla, M\'exico}
\email{jlavalle@cs.buap.mx}

\keywords{counting constraints, satisfiability, query reasoning, XML schemas}
\ACMCCS{{\bf [Theory of Computation]: }Logic---Modal and temporal logics; 
Logic---Automated reasoning;
Formal languages and automata theory---Tree languages;
{\bf [Information systems]: }Data management systems---Query languages---XML query languages---XPath}

\begin{abstract}
We introduce a logical foundation to reason on tree structures with constraints on the number of node occurrences.
 Related formalisms are limited to express occurrence constraints on particular tree regions, 
as for instance the children of a given node.
By contrast, the logic introduced in the present work can concisely express numerical bounds on any region, descendants or ancestors for instance. 
We prove that the logic is decidable in single exponential time even if the numerical constraints are in binary form.

We also illustrate the usage of the logic in the description of numerical constraints on multi-directional path queries on XML documents. 
Furthermore, numerical restrictions on regular languages (XML schemas) 
can also be concisely described by the logic.
This implies a characterization of decidable counting extensions of XPath queries 
and XML schemas.
Moreover, as the logic is closed under negation,
 it can thus be used as an optimal reasoning framework  for testing  emptiness, containment and equivalence.
\end{abstract}

\maketitle

\section{Introduction}
XML is nowadays recognized as the standard technology in the description and exchange of data in the World Wide Web.
One of the cornerstones in the XML community is XPath,
which has been well-established as the most accepted query language for XML documents (finite unranked trees). XPath takes also an important role in other XML technologies, such as
XSLT, XProc and XQuery.
The navigational core of XPath is formed by regular path queries, and its expressive power 
corresponds to the first order logic with two variables FO$^2$ \cite{Marx05}.
A regular path query selects the nodes obtained by the navigation of the path. 
Consider for instance the following query:
$\ifc^\star:a/\fc:b$.
This query expression navigates through the ancestors nodes ($\ifc^\star$) named $a$, and from there it selects the children ($\fc$) labeled with $b$.
The XPath language specification \cite{xpath} also defines arithmetical constructs on the number of node occurrences,
for example: 
$\fc:c\left[\fc^\star:\! a >\; \ifc^\star:\! b\right]$.
This query  selects the $c$ children with more descendants named $a$  than  ancestors named $b$.
However, extending regular path queries with arithmetical constructs leads to undecidability \cite{CateM09}.
Here we focus our study on numerical constraints, 
that is, restrictions with respect to constants (in binary), as for instance:
$\fc:c\left[\fc^\star:\! a >\; 5\right]$.
In this query, the selection is constrained to the $c$ children with more than $5$  descendants named $a$.
In this paper, we identify decidable extensions of XPath with numerical constraints on any regular path. 


Query reasoning in the presence of XML schemas is
one of the central issues that arises from the static analysis of XML specifications and transformations.
XML schemas are used to describe sets of trees by means of regular expressions.
Regular tree languages (types) subsume most XML schema languages used in practice, 
such as XML schema, DTDs and RelaxNG  \cite{MurataLMK05}.
Numerical constraints on regular languages are widely used in many technologies,
such as egrep \cite{Hume88}, Perl \cite{DBLP:books/daglib/0001914} and XML schema languages \cite{MurataLMK05}.
These constraints serve to bound the number of occurrences.
For instance, the regular language over $\{a,b\}$, such that $a$ occurs exactly once and $b$ occurs at least four times, can be written as follows: 
$$(abbbb^+)\mid (babbb^+) \mid (bbabb^+) \mid (bbbab^+) \mid (bbbb^+ a)$$
However, in general, hardcoding numerical constraints  produces  exponentially larger expressions than the original problem \cite{Gelade10}.  
This implies a drastic impact in the computational cost of reasoning on these kind of constraints, 
more precisely, reasoning on hardcoded numerical constraints is exponentially more expensive.
Furthermore, Gelade \cite{Gelade10} also showed that even if the numerical constraints are
directly translated to NFAs, the exponential blow-up cannot be avoided.
In this paper, we provide a way to avoid this exponential blow-up by a succinct characterization
of regular languages with numerical constraints.
More precisely, in the current work it is proposed a tree logic with counting constructs.
These constructs can restrict the number of node occurrences with respect to a constant
coded in binary. It is also shown that the proposed logic is decidable in exponential time.
Also, we show that regular tree expressions (and queries) with counting operators can  be linearly embedded by the proposed logic.


\subsection*{Motivations and Related Work}
The fully enriched $\mu$-calculus is the modal logic with inverse and graded modalities, nominals, a least and a greatest fixed-points. 
Graded modalities are used to constrain the number
of immediate successors of certain node with respect to a constant.
The fully enriched $\mu$-calculus was shown to be undecidable by Bonatti et al. \cite{BonattiLMV06}.
 Nevertheless, it has been recently shown that this result does not apply in the context of finite trees; more precisely, it was provided in \cite{BarcenasGLS11} a single exponential satisfiability algorithm for the fully enriched $\mu$-calculus for trees.
However, graded modalities (in trees)  are limited to impose numerical bounds on the number of children nodes only. 
Although, it was shown in \cite{BiancoMM10} that numerical constraints on descendant nodes
can be expressed by graded $\mu$-calculus formulas, this comes at an exponential cost 
in the formula size. This implies that, even at the logical level, hardcoding in-depth numerical constraints produces an exponential blow-up.
In contrast, we show in this paper, that
our logic can express descendant constraints without an extra cost with respect to the $\mu$-calculus.
In addition, backward constraints, such as on ancestor nodes, can also be expressed for free. 


Seidl et al. \cite{SeidlSM03} showed that the extension of monadic second order logic (MSOL) with Presburger arithmetic is undecidable.
In other works \cite{DemriL10,Dal-ZilioLM04,SeidlSMH04},
 decidable extensions of tree logics with Presburger arithmetical constraints on children are broadly studied.
Demri and Lugiez \cite{DemriL10} provide a PSPACE bound on the decidability of 
modal logic extended with Presburger constraints on children nodes.
 When proving decidability of a fragment of ambient logic,
Dal-Zilio et al. \cite{Dal-ZilioLM04} introduced a modal tree logic with Presburger arithmetic and regular constraints.
In an independent work, Seidl et al. \cite{SeidlSMH04} introduced a decidable extension 
to the logic of Dal-Zilio et al. \cite{Dal-ZilioLM04}.
The extension consists of a fixed-point operator.

In this paper we choose a different trade-off, 
we propose a tree logic with less general cardinality constraints (with respect to binary constants) 
on more extensive tree regions (descendants, ancestors, etc.).
In the same vein,
 it has been recently proposed Bianco et al. \cite{BiancoMM09} a graded version of the computation tree logic CTL .
This logic can pose constraints on the number of paths expressed by CTL formulas.
Constraints are made with respect to constants written in unary form. 
In \cite{BiancoMM10}, the same result was later extended  with constants coded in binary.
This approach however does not support backward navigation, neither in the graded formulas, nor in the non-graded ones.
One consequence is that cardinality constraints can only be expressed on downward tree regions, 
as children or descendants of a given node. 
It should also be recalled that CTL is not as expressive as MSOL. 
This implies that some regular properties, as the ones in XML schemas, cannot be expressed by CTL formulas.
Besides expressing numerical constraints on any multi-directional regular path, 
our logic is as expressive as MSOL and can concisely capture regular tree languages (XML schemas). 

The notion of global constraints has been also subject of recent study in \cite{DBLP:conf/lics/BargunoCGJV10,DBLP:journals/corr/abs-1302-6960}.
Burgo\~no et al. \cite{DBLP:conf/lics/BargunoCGJV10,DBLP:journals/corr/abs-1302-6960} introduced an automata model capable to test (dis)equality 
modulo a given flat equational theory.
 In addition, global numerical constraints (with respect to constants)
 can also be tested.
It is proven emptiness decidability without a further complexity analysis.
In this paper, besides showing decidability of a logic resulting from the addition
of global numerical constraints to a alternation-free two-way $\mu$-calculus for trees,
 we  provide an optimal satisfiability algorithm for the logic.



\subsection*{Contributions and Outline}
We introduce in Section~\ref{sec-logic} an extension of the $\mu$-calculus (for trees) with global counting constructs called $\mu$TLIN.
These constructs restrict the number of nodes (with respect to binary constants) occurring in any region of the tree models. 

In Section~\ref{sec-paths}, we describe a useful application of $\mu$TLIN in the context of XML.
It is shown that an extension of XPath with counting constructs on multi-directional regular paths can be linearly embedded by
the logic. 

Analogously as in Section~\ref{sec-paths}, we  provide in Section~\ref{sec:types} a linear embedding for regular tree languages (XML schemas) with counting constructs.

Section \ref{sec:succincteness} is about succinctness. 
It is shown that the logic with global constraints
is at least exponentially more succinct than the graded $\mu$-calculus.

Section~\ref{sec:decidability} is devoted to show that the proposed logic is decidable.
With this result we can thus use the logic as a reasoning framework for 
XPath queries with schema and counting constraints.
 However, the time complexity bound set for decidability is doubly exponential.

We improve the complexity bound for the logic  in Section~\ref{sec:sat}.
It is described 
a satisfiability algorithm for the logic, and it is  shown that the time complexity of the algorithm is single exponential. Before the description of the algorithm, we provide some
preliminaries in Section~\ref{sec-trees}.
The complexity bound for the satisfiability algorithm, together with the linear
embedding in Sections~\ref{sec-paths} and \ref{sec:types}, 
imply EXPTIME characterizations  of regular path queries (XPath) 
and regular tree languages (XML schemas) extended with global numerical constraints.
Moreover, due to the fact that reasoning on regular tree languages is in EXPTIME-complete,
the logic then represents an optimal reasoning framework for XPath queries and XML schemas with counting.

We conclude in Section~\ref{sec-conclu} with a summary of the paper and a discussion of further  research directions.

\section{A Modal Tree Logic with Global Numerical Constraints} \label{sec-logic}
We consider through the paper labeled unranked trees.
The tree logic with global numerical constraints ($\mu$TLIN) is a modal tree logic (TL) with a least fixed-point ($\mu$), inverse modalities (I), and global numerical constructs (N).
In contrast with graded modalities, where the number of nodes can be restricted only if they are immediate successors of a given node, the counting constructs in our logic can restrict the number of nodes occurring in any part of the tree model.

\subsection{Syntax and semantics}

In the context of tree models, modalities $\m$ in modal formulas are defined by $\M=\{\fc,\fs,\ifc,\ifs\}$.
$\fc$ and $\fs$ stand for the children and right sibling relations, respectively. $\ifc$ and $\ifs$ are the corresponding inverse modalities, that is, the parent and left sibling relations. For a modality $\m$, its inverse is written $\dual{\m}$.

\begin{defi}[Syntax]
We define the set of $\mu$TLIN formulas  with the following grammar:
\begin{align*}
\f :=&  \p \mid \x   \mid \neg \f  \mid \f \vee \f \mid \mf{\m}{\f} \mid \muf{\x}{\f} \mid  \cf{\f}{> k}   
\end{align*}
\end{defi}
Numerical constraints $k$ in counting formulas are assumed to be  integer numbers in binary form.
We use the following notation: 
$\f \wedge \ff$ instead of $\neg(\neg \f \vee \neg \ff)$, 
$\top$ instead of $\f\vee \neg\f$, and
$\cf{\f}{\leq k}$ instead of $\neg (\cf{\f}{> k})$.
In the sequel, we often write counting formulas $\cf{\f}{\# k}$ for $\#\in\{\leq,>\}$.
We define the size (length) of a formula $|\f|$ as usual:
$|p|=|x|=1$; $|\neg \f|= |\mf{\m}\f|=|\muf{x}\f|=1+|\f|$; 
$|\f \vee \ff|= 1+ |\f| + |\ff|$; and $|\cf{\f}{> k}|= \log{(k+1)} + |\f|$.

We consider the traditional assumption that variables can only occur in the scope of a modality or a counting operator.
In addition, we assume variables do not occur in the scope of both, a modality and its converse.
For instance, $\muf{x}{\mf{\fc}{x} \vee \mf{\ifc}{\x}}$ is not 
allowed\footnote{If variables do not occur in the scope of both, a modality and its converse, the greatest and least fixed-points coincide in the context of finite trees \cite{GenevesLS07}.}. 


In a given tree, formulas are interpreted as subsets of tree nodes. 
Propositions serve as node labels.
Negation is interpreted as set complement.
Conjunctions and disjunctions are interpreted as the intersection and union of sets, respectively.
Modal formulas $\mf{\m}{\f}$ are true in a node when there is an accessible node, through $\m$,
such that the formula $\f$ holds. 
The $\mu$ operator is interpreted as a least fixpoint.
The formula $\cf{\f}{> k}$ holds in every node of the tree model, if and only if,
$\f$ holds in {\em at least $k+1$} nodes in the {\em entire tree} (see Definition~\ref{def:sem}).



We now give a formal description of the formula semantics.
Finite tree structures are defined in the style of Kripke transition systems.
\begin{defi}[Trees]
A tree structure, or simply a tree, is a tuple $T=(\Prop,\Nodes,\Rel,\Lab)$, such that:
\begin{itemize}
\item $\Prop$ is the set of propositions;
\item $\Nodes$ is the finite set of nodes; 
 \item $\Rel$ is a transition relation $\left(\Nodes\times \M\right) \times \Nodes$ ($\M$ is the set of modalities) forming a tree structure, 
we write $\nod^\prime\in \Rel(\nod,\m)$ when $(\nod,\m,\nod^\prime)\in \Rel$;
and 
\item $\Lab$ is a left-total labeling relation on $\Nodes\times \Prop$, written $p\in \Lab(n)$.
\end{itemize}
\end{defi}

\begin{defi}[Semantics] \label{def:sem}
Given a tree $T$ and a valuation $V:Var\mapsto 2^{\Nodes}$, where $Var$ is a fixed set of variables, the formula semantics is defined as follows:
\begin{align*}
& \sem{\p}{T}{V} = \{\nod\mid\p\in\Lab(\nod)\}
&& \sem{x}{T}{V}=V(x)\\
 &\sem{\neg \f}{T}{V} = \Nodes\setminus \sem{\f}{T}{V} 
&& \sem{\f \vee \ff}{T}{V} = \sem{\f}{T}{V} \cup \sem{\ff}{T}{V} \\
& \sem{\mf{\m}{\f}}{T}{V} = \{\nod\mid \Rel(\nod,\m)\cap\sem{\f}{T}{V}\neq \emptyset\}
&& \sem{\muf{\x}{\f}}{T}{V} = \bigcap \left\{\Nodes^\prime \mid \sem{\f}{T}{V\subst{\Nodes^\prime}{x}} \subseteq \Nodes^\prime\right\} \\
& \sem{\f> k}{T}{V} = \left \{ \begin{array}{ll} \Nodes & \text{if }|\sem{\f}{T}{V}|> k \\ \emptyset &\text{otherwise} \end{array} \right.
\end{align*}
\end{defi}

\noindent If the interpretation of a formula $\f$ is not empty for a given tree $T$, i.e. $\sem{\f}{T}{V}\neq \emptyset$, 
we say the tree $T$ satisfies the formula $\f$. This is often written $T\models \f$. 
A formula is said to be satisfiable if there is a tree satisfying it. 
Two formulas $\f$ and $\ff$ are equivalent, if and only if, for every tree $T$, $T$ satisfies $\f$, if and only if, $T$ satisfies $\ff$.

\begin{exa}
We can express existential statements with counting formulas. For instance, if we want to select the nodes expressed by a formula $\ff$, only if {\em  there is a node} satisfying $\f$, then we write:
$$(\cf{\f}{> 0}) \wedge \ff$$ 
Universality can also be expressed. The following formula selects the $\ff$ nodes when {\em every node} satisfies $\f$:
$$[\cf{(\neg \f)}{\leq 0}] \wedge \ff$$ 
Note that with counting formulas it is also possible to restrict the number of nodes occurring in a particular region.
First, consider for instance the {\em descendants region}. This can be expressed as follows:
$$\muf{\x}{ \mf{\ifc}{(\p_0 \vee \x)}}$$ 
This formula denotes {\em the descendants of the $\p_0$ nodes}.
Recall that $\ifc$ denotes the parent relation. 
Hence, the formula holds in nodes from where,
by  recursive navigations through   parents,
nodes named $\p_0$ are accessible.
Then, if we  want to restrict the number of descendants of the $\p_0$ nodes in a tree, then we write:
$$\cf{[\muf{\x}{ \mf{\ifc}{(\p_0 \vee \x)}}]}{\leq 6}$$
 Now, if we want to restrict 
the number of some descendants, say descendants named $\p_1$,
then we write:
	$$\cf{([\muf{\x}{ \mf{\ifc}{(\p_0 \vee \x)} } ]\wedge \p_1)}{\leq 6}$$
Notice that $\muf{\x} \mf{\ifc}(\p_0 \vee \x) \wedge \p_1$
holds in {\em all} $\p_1$ descendants of {\em each} $\p_0$ node.
Hence, if in a model there are $2$ nodes named $\p_0$ with $2$ and $3$ descendants named $\p_1$, respectively,
then the formula $\cf{\left([\muf{\x}{ \mf{\ifc}{(\p_0 \vee \x)} \vee \mf{\ifs}{\x} }] \wedge \p_1\right)}{\leq 6}$ holds
due to all $6$ descendants of both $\p_0$ nodes (see Figure~\ref{fig:example:formulas}).
However,
one may also want to restrict the number of descendants of a particular node.
This can be done by isolating the origin node from where navigation starts (during counting).
For this purpose we first define the following formula:
$$(\cf{\nom}{\leq 1} )\wedge (\cf{\nom}{>0})$$
In this formula, proposition $\nom$ occurs exactly once in a model.
If we want to indentify where $\nom$ occurs, then we write:
$$(\cf{\nom}{= 1})\wedge \nom$$
where $\cf{\nom}{= 1}$ stands for $(\cf{\nom}{\leq 1}) \wedge (\cf{\nom}{>0})$.
Note that formula $(\cf{\nom}{= 1})\wedge \nom$ selects an node only if
the formula is true in exactly that node, then this formula can be seen as a  
nominal~\cite{BonattiLMV06}.
Now that we can isolate a single node in a model, we can thus restrict the counting
from a particular node, consider for instance the following formula:
$$\cf{[\muf{\x}{ \mf{\ifc}{([(\cf{\nom}{=1} )\wedge \nom \wedge \p_0] \vee \x)}}]}{\leq 2}$$
This formula is true in models where there is single node with no more than
$2$ descendants. If in addition, we want to name the descendants, say $p_1$,
then we write:
$$\cf{[\muf{\x}{ \mf{\ifc}{([(\cf{\nom}{=1} )\wedge \nom \wedge \p_0] \vee \x)}} \wedge \p_1]}{\leq 2}$$
A graphical respresentation of the examples above is depicted in Figure~\ref{fig:example:formulas}.


\begin{figure}[t]
\begin{center}
\begin{tikzpicture}[level/.style={sibling distance=15mm, level distance = 10mm}]
\node [circle,draw,inner sep=1pt] (z){$p_2$}
	child {node [circle,draw,inner sep=1pt] (a) {$p_0$} 
		child {node [circle,draw,inner sep=1pt] (a1) {$p_1$}	
			child [grow=up,level distance = 5mm] {node (r) {$\phi$} edge from parent[draw=none]}
		}
		child [grow=up,level distance = 5mm] {node (r) {$\nom$} edge from parent[draw=none]}
		child {node [circle,draw,inner sep=1pt] (a2) {$p_2$}
			child [grow=up,level distance = 5mm] {node (r) {$\phi$} edge from parent[draw=none]}
			child {node [circle,draw,inner sep=1pt] (a3) {$p_1$}
				child [grow=up,level distance = 5mm] {node (a4) {$\phi$} edge from parent[draw=none]}
			}
		}
	}
	child [grow=up,level distance = 5mm] {node (r) {} edge from parent[draw=none]}
	child [grow=up,level distance = 5mm] {node (r) {} edge from parent[draw=none]}
	child {node (r2) {} edge from parent[draw=none]}
	child {node [circle,draw,inner sep=1pt] (v) {$p_0$} 
		child {node [circle,draw,inner sep=1pt] (v1) {$p_1$}
			child [grow=up,level distance = 5mm] {node (r) {$\phi$} edge from parent[draw=none]}
		}
		child {node [circle,draw,inner sep=1pt] (v1) {$p_2$}
			child {node [circle,draw,inner sep=1pt] (v2) {$p_1$}
				child [grow=up,level distance = 5mm] {node (r) {$\phi$} edge from parent[draw=none]}
			}
			child [grow=right,level distance = 5mm] {node (r) {$\phi$} edge from parent[draw=none]}
		}
		child {node [circle,draw,inner sep=1pt] (v1) {$p_1$}
			child [grow=up,level distance = 5mm] {node (r) {$\phi$} edge from parent[draw=none]}
			child {node [circle,draw,inner sep=1pt] (v2) {$p_1$}
				child [grow=up,level distance = 5mm] {node (r) {$\phi$} edge from parent[draw=none]}
			}
		}
	}
;
\end{tikzpicture}
\end{center}
\caption{Tree model example: descendant region of $p_0$ nodes is denoted by the formula $\f\equiv \muf{\x}{ \mf{\ifc}{(\p_0 \vee \x)}} $; 
formula $\cf{(\phi \wedge \p_1)}{\leq 6}$ holds because 
the $\p_0$ nodes have  exactly $6$  descendants labeled with $\p_1$;
$\cf{(\muf{\x}{[(\p_0\wedge \nom \wedge \nom=1) \vee \mf{\ifc}\x]} \wedge \p_1 )}{\leq 2}$ is true
because there is a $\p_0$ node, the one marked with $\nom$,
with  $2$ descendants named $\p_1$. 
 }
\label{fig:example:formulas}
\end{figure}
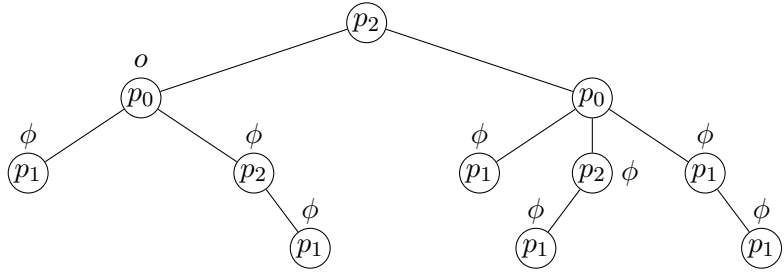
\end{exa}



\section{Counting Regular Path Queries} 
\label{sec-paths}
The navigation core of the XPath query language (for XML documents)
has been formalized as regular path queries,
and it is known to correspond to FOL$^2$ \cite{CateM09,Marx05}.
In this Section, we introduce an extension of regular path queries
with counting constructs. 
In contrast with the counting extension of regular path queries reported in \cite{BarcenasGLS11},
where counting is limited to children paths only,
the counting constructs described in this work are able to constrain arbitrary regular paths.
We also provide in this Section a linear characterization of the counting extension of regular path queries into $\mu$TLIN.


\subsection{Syntax and semantics}
We first describe the extension of regular paths with  counting constructs on multi-directional paths. 
We call this extension CPath.
\begin{defi}[Syntax]
The syntax of CPath queries $\rho$ is given as follows:
\begin{align*}
 \alpha :=& \fc \mid \fs \mid \ifc \mid \ifs \mid \fc^\star \mid \ifc^\star &
 \varrho:=& \top \mid \alpha \mid \p \mid \alpha:\p \mid \varrho/ \varrho \mid \varrho[\beta]\\
\beta := & \varrho > k  \mid \beta \vee \beta \mid \neg \beta &
\rho := & \varrho \mid /\rho \mid \rho \cup \rho \mid \rho \cap \rho \mid \rho \setminus \rho
\end{align*}
where $\p$ is a proposition, and $k$ is a positive integer in binary. 
\end{defi}
We also consider the following syntatic sugar:
$\cf{\varrho}{\leq k}$ is written instead of $\neg (\cf{\varrho}{> k})$; 
$\varrho$  instead of $\cf{\varrho}{> 0}$;
$\beta_1\wedge \beta_2$ instead of $\neg\left(\neg\beta_1 \vee \neg\beta_2\right)$;
and $\varrho[\beta_1][\beta_2]$ instead of $\varrho[\beta_1\wedge \beta_2]$.

The CPath expressions are interpreted as node-selection queries on tree structures.
In particular,
the axis relations $\alpha$ are interpreted as follows: children $\fc$, following sibling $\fs$,  parent $\ifc$, previous sibling $\ifs$,
descendants $\fc^\star$, and ancestors $\ifc^\star$. Step paths $\alpha:\p$ selects the $\p$ nodes reachable by $\alpha$.
Symbol $/$ is used to compose paths. A qualified path $\varrho[\beta]$ selects the nodes denoted by $\varrho$ that satisfies the boolean condition $\beta$. A qualified path $[\varrho > k]$ is true when $\varrho$ selects at least $k$ nodes.
The boolean combination of qualifiers $\beta$ are interpreted in the obvious manner.
The path $/\rho$ selects the nodes denoted by $\rho$ that are reachable from the root.
Union, intersection and difference of paths are interpreted as expected. 
Before given a formal description of the CPath semantics (inspired from \cite{CateM09}), 
we introduce the following notation: 
in a Kripke structure, $\nod_1\stackrel{\alpha}{\rightarrow} \nod_2$ means
than $\nod_1$ is related by means of $\alpha$ with $\nod_2$,
where $\alpha$ can be any axis relation ($\fc,\fs,\ifc,\ifs,\fc^\star,\ifc^\star$).

\begin{defi}[Semantics] \label{def:pathsemantics}
The semantics of CPath queries  is defined by a function $\sem{\cdot}{\cdot}{}$
from CPath queries with respect to a tree $T$, to pairs of nodes in $T$.
\begin{align*}
&\sem{\top}{T}{} = \Nodes\times\Nodes 
&&\sem{\p}{T}{} =\{(\nod,\nod)\mid \p\in\Lab(\nod)\}\\
&\sem{\alpha}{T}{} = \{(\nod_1,\nod_2)\mid \nod_1\stackrel{\alpha}{\rightarrow} \nod_2\}
&&\sem{\alpha:\p}{T}{} = \{(\nod_1,\nod_2)\in\sem{\alpha}{T}{}\mid \p\in\Lab(\nod_2)\}\\
&\sem{\varrho_1/\varrho_2}{T}{} = \sem{\varrho_1}{T}{}\circ\sem{\varrho_2}{T}{}
&&\sem{\varrho[\beta]}{T}{} = \{(\nod_1,\nod_2)\in\sem{\varrho}{T}{} \mid  \nod_2\in\semm{\beta}{T}{} \}\\
&\semm{\varrho > k}{T}{} = \{\nod_1\mid |\{\nod_2\mid (\nod_1,\nod_2)\in \sem{\varrho}{T}{}\}|>k\}
&&\semm{\neg\beta}{T}{} = \Nodes\setminus \semm{\beta}{T}{} \\
&\semm{\beta_1\vee \beta_2}{T}{} =  \semm{\beta_1}{T}{} \cup \semm{\beta_2}{T}{}
&&\sem{/\varrho}{T}{} = \{(r,\nod)\in \sem{\varrho}{T}{} \mid r\text{ is the root}\} \\
&\sem{\rho_1\cup\rho_2}{T}{}= \sem{\rho_1 }{T}{} \cup \sem{\rho_2}{T}{}\
&&\sem{\rho_1\cap\rho_2}{T}{}= \sem{\rho_1 }{T}{} \cap \sem{\rho_2}{T}{} \\
&\sem{\rho_1\setminus\rho_2}{T}{}= \sem{\rho_1 }{T}{} \setminus \sem{\rho_2}{T}{}
\end{align*}
Notice that the function $\semm{\cdot}{\cdot}{}$ is introduced to distinguish the interpretation of paths inside qualifiers.
\end{defi}

\begin{exa}
Consider for instance the following composition of paths:
$$\ifc^\star:p_1/\fc^\star:p_2$$
This query, evaluated from some context (a node subset), navigates to the $\p_1$ ancestors
of the context, and from there, it selects the $p_2$ descendants.
Now consider the following qualified path:
$$\ifc^\star:p_1[\fc^\star:p_2]$$
In constrast with the previous example, this query selects the $\p_1$ ancestors {\em with at least $1$ descendant named $p_2$}.
\end{exa}

\begin{prop}[Succinctness]
For any tree $T$ and CPath expression $\rho$, there is a regular path (CPath without counting) $\rho^\prime$, such that
\begin{itemize}
\item $\sem{\rho}{T}{}=\sem{\rho^\prime}{T}{}$, and
\item the size of $\rho^\prime$ is exponentially greater than the size of $\rho$.
\end{itemize}
\end{prop}
\begin{proof}
Given an expression $\varrho> k$ ($k>0$), we will show that
that there is an equivalent path expression $\varrho^\prime$  without counting.
We proceed by induction on the structure of $\varrho$.

For the base cases, we consider the following replacements:
\begin{align*}
  \fc:\p[\overbrace{\fs:\p[\fs:\p[\ldots]]}^{\text{$k$ times}}]
&\text{ instead of } \fc:\p>k;\\
\overbrace{\fs:\p[\fs:\p[\ldots]]}^{\text{$k+1$ times}}]
&\text{ instead of } \fs:\p>k;\\
\overbrace{\ifs:\p[\ifs:\p[\ldots]]}^{\text{$k+1$ times}}]
&\text{ instead of } \ifs:\p>k;\\
\fc^\star:p[\overbrace{\fc^\star\fs:p[\fc^\star\fs:p[\ldots]]}^\text{$k$ times}]
&\text{ instead of } \fc^\star:p>k; \text{ and}\\
\ifc^\star:p[\overbrace{\ifc^\star\fs:p[\fc^\star\fs:p[\ldots]]}^\text{$k$ times}]
&\text{ instead of } \ifc^\star:p>k;
\end{align*}
where $\fc^\star\fs:p[\beta]$ and $\ifc^\star\fs:p[\beta]$ are syntactic sugar for 
$\fc^\star:p[\beta]\; \vee \fs:p[\beta]$ and $\ifc^\star:p[\beta]\; \vee \fs:p[\beta]$,
respectively.

Consider now the case $(\varrho_1/\varrho_2)>k$, this expression
is replaced by $\varrho_1^\prime/\varrho_2^\prime$,
where by induction we know that $\varrho_1^\prime$and $\varrho_2^\prime$ 
are the counting-free  expressions equivalent to $\varrho_1$ and $\varrho_2>k$,
respectively.

Expression $\varrho_1[\varrho_2 > k_2]>k_1$ is replaced by
$\varrho_1^\prime[\varrho_2^\prime]$, such that by induction
$\varrho_1^\prime$ and $\varrho_2^\prime$ are the counting-free expressions
equivalent to $\varrho_1>k_1$ and $\varrho_2> k_2$, respectively.

Cases $\varrho[\beta_1 \vee \beta_2]$ and $\varrho[\neg\beta]$
are also immediate by induction.

In the replacement described above, notice that numerical restrictions (in binary)
are replaced by explicit path occurrences, 
it is  hence easy to see  the exponential blow-up in the size of the 
counting-free expression.
\end{proof}

\begin{defi}[Reasoning problems]
We define the emptiness, contaiment and equivalence problems of CPath queries
as follows.
\begin{itemize}
\item We say a query $\rho$ is empty, if and only if,
for every tree $T$, its interpretation is empty, that is, $\sem{\rho}{T}{}=\emptyset$;
\item It is said that a query $\rho_1$ is contained in a query $\rho_2$,
if and only if,
for every tree $T$, each pair of nodes in the interpretation of $\rho_1$ is in the intepretation
of $\rho_2$, that is, $\sem{\rho_1}{T}{}\subseteq \sem{\rho_2}{T}{}$; and
\item Two queries $\rho_1$ and $\rho_2$ are equivalent,
if and only if,
for every tree $T$, $\rho_1$ is contained in $\rho_2$ and the other way around,
that is, $\sem{\rho_1}{T}{}\subseteq \sem{\rho_2}{T}{}$ and
 $\sem{\rho_2}{T}{}\subseteq \sem{\rho_1}{T}{}$.
\end{itemize}
\end{defi}

\subsection{Logic characterization}
Regular path queries (without counting) can be written in terms of the $\mu$-calculus \cite{BarcenasGLS11}. 
For instance,  the query $\fc^\star:\p$, evaluated in the root $r$, selects the $p$ descendants of $r$.
 This can be written as follows: $$\left[\muf{\x}{\mf{\ifc}(r\vee\x)}\right] \wedge \p$$
If we want to evaluate the query in another context (node subset), represented by a  $C$ formula, then we simply replace the occurrence of $r$ by $C$. For instance,
 let us say the context represented by all the nodes named $\p_0$, then the $\p$ ancestors of $\p_0$ nodes can be written as follows:
$$\left[\muf{\x}{\mf{\ifc}(\p_0\vee\x)}\right] \wedge \p$$
In \cite{BarcenasGLS11}, it was also shown that an extension of regular path queries with counting on {\em children paths} 
can be expressed in terms of the two-way {\em graded} $\mu$-calculus. Children paths are of the  forms $\fc:\p$ and $\fc:\p[\varrho]$.
In this paper, we show that the $\mu$TLIN  counting constructs can describe more general counting constructs 
on arbitrary regular path queries, such as $\fc^\star:\p_1/\ifc^\star:\p_2[\varrho]$.

\begin{defi}[CPath queries into $\mu$TLIN formulas]
Given a context formula $C$,
the translation $F$ from CPath queries into $\mu$TLIN formulas is defined as follows:
\begin{align*}
 &F(\fc, C)  =\mf{\ifc}{C} 
&& F(\fs, C)  = \mf{\ifs}{C} \\
& F(\ifc,C)  = \mf{\fc}C 
&& F(\ifs,C)  = \mf{\fs}{C} \\
& F(\fc^\star,C) = \muf{\x}{\mf{\ifc}(C\vee\x) } 
&& F(\ifc^\star,C) =\muf{\x}{\mf{\fc}{(C \vee \x})}\\
& F(\alpha:\p,C) = F(\alpha,C) \wedge \p 
&& F(\varrho_1/\varrho_2,C) = F(\varrho_2,F(\varrho_1,C))  \\
& F(\varrho[\beta],C) =F(\varrho,C) \wedge \nom \wedge F(\beta,[\nom\!=\!1]\! \wedge\! \nom) 
&& F(\varrho> k,C) = \cf{F(\varrho, C)}{> k}  \\
& F(\neg \beta,C) = F^\prime(\beta,C) 
&& F(\beta_1\vee\beta_2,C) =F(\beta_1,C) \vee F(\beta_2,C)\\
& F(/\varrho,C) =\! \!F(\varrho,C\!\!\wedge\!\neg(\mf{\ifc}\top\!\wedge\!\mf{\ifs}\top)) 
&&  F(\rho_1\cap\rho_2,C) = F(\rho_1,C)\wedge F(\rho_2,C) \\
& F(\rho_1\cup\rho_2,C) = F(\rho_1,C)\vee F(\rho_2,C) 
&& F(\rho_1\setminus \rho_2,C) =\! F(\rho_1,C)\!\wedge\! F^\prime(\rho_2,C) 
\end{align*}
where
\begin{align*}
F^\prime(\rho) = & \begin{cases} 
		F^\prime( \varrho,C  \wedge \neg[\mf{\ifc}\top  \wedge  \mf{\ifs}\top])
		& \text{ if $\rho$ has the form $/\varrho$}, \\
	 	\neg F(\rho) & \text{ otherwise.}
		\end{cases} \\
F^\prime(\varrho) = & \begin{cases}
	\neg F(\varrho^\prime,C) \vee \left[\nom \wedge \neg F(\beta,[\nom\!=\! 1] \wedge \nom) \right] 
	& \text{ if $\varrho$ has the form $\varrho^\prime[\beta]$}, \\
	\neg F(\varrho) & \text{ otherwise.}
		\end{cases}
\end{align*}
In general $F^\prime$ represent the negation of $F$, 
however in the case where there is a counting operator,
the fresh proposition $\nom$, which serves to fix an origin node, is not negated. 
Note that the constraint $\nom\!=\!1 \wedge \nom$ is not affected by negation because it always
occur in the scope of a counting operator.
\end{defi}

\begin{exa}
Consider the following query  evaluated in a context $C$:
$$\fc:\p_1[\fc^\star:\p_2>k]$$
The query selects the $\p_1$ children of $C$ with at least $k+1$ descendants named $\p_2$.
The first part of the query $\fc:\p_1$ is translated as follows:
$$\p_1\wedge \mf{\ifc}C $$
That is, the $\p_1$ nodes with $C$ as parent.
The translation of the counting expression $\fc^\star:\p_2>k$ is
$$\cf{\nom \wedge \left[\p_2\wedge\muf{\x}\mf{\ifc}([\nom\!=\!1 \wedge \nom] \vee \x)\right]}{>k}$$
This formula holds, if and only if, there are more than $k$ descendant nodes, named $\p_2$, of a single node named $\nom$.
Then, the translation of the entire query is the following:
\begin{align*}
F(\fc:\p_1[\fc^\star:\p_2>k])=
\left(\p_1 \wedge \mf{\ifc}C   \right) \wedge 
\left( \nom \wedge 
\cf{ \left[ \p_2\wedge\muf{\x} \mf{\ifc}\left([\nom\!=\!1 \wedge \nom] \vee \x\right) \right]}>k \right)
\end{align*}
The proposition $\nom$ is used to fix a context for the counting subformula.
$\nom$ holds in a single $\p_1$ node, then the $\p_2$ descendants of that particular $\p_1$ node are the only ones counted.
\end{exa}


With the translation function $F$, we can now use the logic as
a reasoning framework to solve emptiness, containment and equivalence of 
CPath queries, moreover,
since translation $F$ does not introduce duplications,
 it is easy to see that the formula resulting from the translation  has linear size with respect  to the input query. 
\begin{thm}[Query reasoning]\label{theo-paths}
For any CPath queries $\rho,\rho_1,\rho_2$, tree $T$ and  valuation $V$, the following holds:
\begin{itemize}
\item
 $\sem{\rho}{T}{} = \emptyset \;\text{ if and only if }\;\sem{F(\rho,\top)}{T}{V} = \emptyset$;
\item
$\sem{\rho_1}{T}{} \subseteq \sem{\rho_2}{T}{} \;\text{ if and only if }\; 
  \sem{F(\rho_1,\top) \wedge F^\prime(\rho_2,\top) }{T}{V} = \emptyset$; and
\item  $F(\rho,\top)$ has linear size with respect to $\rho$ and 
$F^\prime(\rho_1,\top) \wedge F(\rho_2,\top)$ has linear size with respect to 
$\rho_1$ and $\rho_2$.
\end{itemize}
\end{thm}

\begin{proof}
For the first item, we proceed by structural induction on $\rho$.

In order to proof the case when $\rho$ has the form $\varrho$,
we will proof the following:
 $\varrho$ evaluated in a context $C$ is satisfiable by a tree $T$, if and only if,
$F(\varrho,C)$ is satisfiable by $T$.

Consider $\rho$ is the basic query $\fc^\star:\p$, 
then $F(\fc^\star:\p,C) = \p \wedge \muf{\x} \mf{\ifc}(C \vee \x)$,
which clearly selects exactly the same nodes than $\rho$ evaluated in $C$.
The proof for the cases with the other axes ($\fc,\ifc,\fs,\ifs,\ifc^\star$) is similar. 

Now let the input query be a composition of paths, that is, $\rho$ has the form $\varrho_1/\varrho_2$. Intuitively, $\varrho_1/\varrho_2$ selects the nodes denoted by $\rho_2$ 
evaluated from the nodes satisfying $\varrho_1$, that is, $\varrho_1$ is the context.
That is precisely what it means $F(\varrho_2,F(\varrho_1,C))$.
By induction $F(\varrho_1,C)$ corresponds to $\varrho_1$, and then
also by induction $F(\varrho_2,F(\varrho_1,C))$ corresponds $\varrho_1/\varrho_2$
evaluated in $C$.

Before proving the case when the input query has the form $\varrho_1[\varrho_2 > k]$,
we need first to proof that $\varrho_2>k$ is  satisfiable by $T$,
if and only if, $ F(\varrho_2,\nom=1\wedge \nom)> k$ is  satisfiable by $T$.
 This is achieved by induction on the structure of $\varrho_2$. 
Consider $\varrho_2$ has the form $\fc:\p$. 
Then $F(\fc:\p,\top)= \p \wedge \mf{\ifc}\top$. This formula  selects all the $\p$ children  of the model.
However  according to the semantics of CPath queries (Definition~\ref{def:pathsemantics}), we need to count the $\p$ children of a single  node.
This is achieved by fixing the context with a new fresh proposition $\nom$ occurring only once in the model $\nom=1$. 
Hence $\left[\p \wedge \mf{\ifc}\left([\nom=1]\wedge \nom \right) \right] >k$
is satisfiable by $T$, if and only if, $\fc:\p>k$ is satisfiable by $T$.  
We proceed analogously for the other axes.
For the other cases of $\varrho_2$, that is, 
when $\varrho_2$ is a composition of paths 
($\varrho_2^\prime/\varrho_2^{\prime\prime}$) and a qualified path
($\varrho_2^\prime[\beta^\prime]$), the proof goes straightforward by induction.

Now that we know that $\varrho_2>k$  is satisfiable by $T$, if and only if, 
$F(\varrho_2)>k$ is satisfiable by $T$, and that by induction, 
$\varrho_1$ evaluated in $C$ is satisfiable by $T$, if and only if, $F(\varrho_1,C)$
is satisfiable $T$, we can thus infer that 
$F(\varrho_1,C) \wedge \nom \wedge F(\varrho_2,[\nom=1] \wedge \nom)$
is satisfiable by $T$, if and only if, $\varrho_1[\varrho_2>k]$ is satisfiable in context $C$ by $T$. 
Note that $\nom$ is used to select a single $\varrho_1$ node.

When $\varrho$ has the form $\varrho_1[\beta]$,  the cases when
$\beta$ is a disjunction or a negation are immediate by induction.
In the case of negation, it is important to notice that
the negation of $F(\varrho^\prime,[\nom=1]\wedge \nom)>k$ does not affect the context,
that is, negation never goes inside the formula $[\nom=1]\wedge \nom$.

Consider now the case when the input query has the form $\rho_1\setminus \rho_2$.
The only interesting case is when $\rho_2$ has the form $\varrho_1[\varrho_2 > k]$.
It is easy to see, by induction, that 
$F(\rho_1,C)$ is satisfiable by $T$, if and only if, $\rho_1$ is satisfiable by $T$.
Also by induction we also know that
$\neg F(\varrho_1,C) \vee \left(\nom\wedge \neg F(\varrho_2,[\nom=1]\wedge \nom)\right)$
is satisfiable by $T$, if and only if, $\varrho_1[\varrho_2>k]$ is not satisfiable by $T$.
We can hence conclude that 
$F(\rho_1,C) \wedge \left[\neg F(\varrho_1,C) \vee \left(\nom\wedge \neg F(\varrho_2 [\nom=1]\wedge \nom)\right)\right]$
is satisfiable by $T$, if and only if, $\rho_1\setminus(\varrho_1[\varrho_2>k])$ also does.

The cases when the input query has the forms $\rho_1\cup \rho_2$,
$\rho_1 \cap \rho_2$, and $/\rho_1$ are straightforward by induction.

For the second item, we proceed analogously as in the first item in the case
when the input query has the form $\rho_1\setminus \rho_2$.

The third item is proven immediately by structural induction on the input query
and by noticing that function $F$ does not introduce duplications.
\end{proof}

\section{Regular Tree Languages with Counting} \label{sec:types}
Regular tree language expressions (types, schemas) can be seen 
as the arborescent version of regular expressions.
These expressions are used to describe sets of trees, and
they encompass most common XML schema languages, such as
DTDs, XML schema and RelaxNG \cite{MurataLMK05}.
 Consider for instance the following expression:
$$\p_1[\p_2 ^ \star]$$
This expression is interpreted as the set of trees (XML documents) rooted by $\p_1$
with $0$ or more contiguous children named $\p_2$.

In this paper, we consider an extension of regular tree languages with counting constructs.
These constructs serve to constrain the number of children occurrence.
For example, if one wants to describe the trees rooted by $\p_1$ with at most $5$ children named $\p_2$, one may write:
$$\p_1[\p_2 ^ {\leq 5}]$$

\subsection{Syntax and semantics}
We now give a precise definition of the regular tree types with counting.
\begin{defi}[CTypes syntax]
The syntax of CTypes expressions is defined by:
\begin{align*}
&e:= \epsilon\mid \x  \mid e\cdot e\mid e+e\mid \text{let $\dual{\x}$.$\dual{e}$ in $e$}
	 \mid \p[e^{> k}] \mid \p[e^{\leq k}]
\end{align*}
\end{defi}
We often write $\p[e]$ instead of $\p[e>0]$.
Variables cannot occur free, that is, variables always occur under
the scope of a fixpoint operator.

$\epsilon$ is used for the empty tree. 
Concatenation and alternation are expressed as usual with the respective symbols $\cdot$ and $+$.
The binder is used for recursion.
The Kleene star and other common notation for regular languages are defined as follows:
$e^\star=\text{let $x.(e\cdot x)+\epsilon$ in $\x$}$, $e^+=e\cdot e^\star$, and $e^?=\epsilon+e$.
Counting expressions $\p[e^{\# k}]$ denote the set of trees rooted at $\p$ such that the number of children subtrees matching with $e$ satisfy the numerical constraint $\# k$.
In contrast with other forms of counting in regular tree languages \cite{Gelade10}, 
we do not force the counted nodes to be contiguous siblings.

\begin{defi}[CTypes semantics]
Given a valuation $V$ into trees, the interpretation of CTypes expressions
is given as follows:
\begin{align*}
& \sem{\epsilon}{}{V} = \{\emptyset\} 
&& \sem{\x}{}{V} = V(\x) \\
& \sem{e_1\cdot e_1}{}{V}= \sem{e_1}{}{V} \cdot \sem{e_2}{}{V}
&& \sem{e_1 + e_2}{}{V} = \sem{e_1}{}{V} \cup \sem{e_2}{}{V}\\
& \sem{\text{let $\dual{\x}.\dual{e}$ in $e$}}{}{V} = \sem{e}{}{\text{lfp($V$)}}
&& \sem{\p[e^{\# k}]}{}{V} = \{T \mid \text{the root of $T$ is labeled by $\p$ and the number}\\ 
&&& \;\;\;\;\; \;\;\;\;\;\;\;\;\;\;\;\;\;\;\;    \text{of children subtrees in $\sem{e}{}{T}$  satisfies \#k} \}
\end{align*}
where lfp($f$) is the least fixpoint of $f$ defined 
$\text{lfp}(V^\prime)=V\left[^{\dual{x}}/_{ \dual{ \sem{e}{}{V^\prime} } }\right]$.
Note that $V$ is monotone according to subset ordering,
 hence it always has a fixpoint due to the Fixpoint Theorem \cite{Tarski55}.
\end{defi}

 It was shown in \cite{BarcenasGLS11} that any CTypes expression can be written in terms
of $\mu$-calculus formulae. 
We also know that the graded $\mu$-calculus is as expressive as
the plain $\mu$-calculus \cite{barcenasthesis}. 
It is also well-known that the plain $\mu$-calculus and regular tree languages (types)
are equally expressive \cite{DBLP:conf/concur/JaninW96}.
It is then easy to see that counting operators (CTypes) do not introduce more
expressive power in regular tree languages. 
Also, by  Theorem~\ref{theo:expressiveness},
we can  conclude that 
$\mu$TLIN and CTypes are equally expressive.

\subsection{Logic characterization}

CTypes without counting can be linearly characterized by the simple $\mu$-calculus \cite{BarcenasGLS11}. Moreover, in the same work it is also shown that the counting constructs of CTypes can be captured by the graded $\mu$-calculus.
We now show that $\mu$TLIN can also capture CTypes expression and hence be used as a reasoning framework.
For instance, the above example $\p_1[\p_2^\star]$ can be expressed as follows:
$$\p_1 \wedge \left( \neg\mf{\fc}\top \vee 
 \mf{\fc}\left[\neg\mf{\ifs}\top \wedge \muf{\x}\p_2 \wedge (\mf{\fs}\x \vee \neg \mf{\fs}\top) \right]\right) $$

We now give a general translation function.
\begin{defi}[CTypes expressions into $\mu$TLIN formulas]
The translation function $F$ from CTypes expressions to $\mu$TLIN formulas
is given as follows:
\begin{align*}
& F(\epsilon) = \neg \top
&& F(e_1+e_2)= F(e_1) \cup F(e_2) \\
& F(e_1\cdot e_2) = F(e_1) \wedge \mf{\fs} F(e_2)
&& F(\text{let $\dual{x}.\dual{e}$ in e}) = \muf{\dual{x}}\dual{F(e)}\text{ in }F(e) \\
\end{align*}
\vspace{-1cm}
\begin{align*}
& F(\p[e^{\# k}])= \p \wedge \nom \wedge 
\cf{\left(F(e)\wedge \mf{\ifc}\left[\nom\!=\! 1 \wedge \nom\right] \right)}{\# k}
\end{align*}
\end{defi}

Formula $\muf{\dual{x}}\dual{\phi} \text{ in }\phi$ is a generalization of
the least fixpoint. 
Its formal semantics is defined as follows:
$$ \sem{\muf{\dual{x}}\dual{\phi} \text{ in }\phi}{T}{V}= 
\sem{\phi}{T}{V\subst{\dual{N^{\prime\prime}}}{\dual{\x}}},
\text{ where }
\dual{N^{\prime\prime}}=\bigcap \left\{\dual{N^\prime} \mid \dual{\sem{\phi}{T}{V\subst{\dual{N^\prime}}{\dual{x}}}} \subseteq \dual{N^\prime} \right\}.$$
Note that this generalization does not provide more expressive power
and it is only used for a succinct translation of his analogous operator in CTypes expressions.


Now consider an example for the translation function.
The expression above $\p_1[\p_2^\leq 5]$ is  translated as follows:
$$F(\p_1[\p_2^{\leq k}])= \; \p_1 \wedge \nom \wedge \cf{\left(\p_2\wedge \mf{\ifc}\left[\nom\!=\!1 \wedge \nom\right] \right)}{\leq 5} $$

Notice that the fresh proposition $\nom$ is used to count from a fixed context
in an analogous manner as done for regular path queries.
It is then necessary to define a safe negation for the translation $F$ in order
to properly model the containment and equivalence of CTypes expressions.
Safe negation of $F$ is defined by $F^\prime$ as follows.
\begin{defi}
We define the following translation function from CTypes expressions into
$\mu$TLIN formulas.
$$
F^\prime(e) = \begin{cases}
				\neg \p \vee \left(\nom \wedge 
				\neg \left[\cf{\left(F(e_0)\wedge \muf{\x}\mf{\ifc}\left[\nom\!=\! 1 
				\wedge \nom\right] \vee \mf{\ifs}\x\right)}{\# k} \right]\right)
           			& \text{ if $e$ has the form }\p[e_0^{\# k}],\\
				\neg F(e) & \text{ otherwise.}
			\end{cases}
$$
\end{defi}

We can now define the reasoning problems of CTypes expressions in terms of
$\mu$TLIN formulas.
\begin{thm}[CTypes reasoning]
\label{theo-types}
For any CTypes expressions $e$, $e_1$ and $e_2$, tree $T$ and valuation $V$,
we have that:
\begin{itemize}
\item $\sem{e}{}{V} = \emptyset$, if and only if, 
$\sem{F(e)}{T}{V} = \emptyset$;
\item $\sem{e_1}{}{V} \subseteq \sem{e_2}{}{V}$, if and only if,
$\sem{F(e_1) \wedge  F^\prime(e_2)}{T}{V} = \emptyset$; and
\item $F(e)$, $F(e_1)$ and $F^\prime(e_2)$ have linear size with respect to $e$, $e_1$
and $e_2$, respectively.
\end{itemize}
\end{thm}
\begin{proof}
The proof goes by structural induction on the input CTypes expressions
in an analogous manner as the proof of Theorem~\ref{theo-paths}.
We will only show the case when the CTypes expression
has the form $\p[e^{\# k}]$ for the first item.
By induction we know $F(e)$ is satisfiable by a tree $T$, if and only if, $e$ is satisfiable.
Then the formula $[F(e) \wedge \mf{\ifc}(\nom=1\wedge \nom) ]\#k$
is satisfiable by $T$, if and only if, there is a node with children matching $F(e)$
and satisfying the numerical constraint $\# k$.
Therefore $\p \wedge \nom \wedge [F(e) \wedge \mf{\ifc}(\nom=1\wedge \nom) ]\#k$
is satisfiable by $T$, if and only if, $\p[e^{\#k}]$ is satisfiable by $T$.
\end{proof}

\section{Succincteness} \label{sec:succincteness}
We show in this Section that $\mu$TLIN is at least exponentially more
succinct that the graded $\mu$-calculus \cite{BonattiLMV06}.
This is done via a GCTL embedding. 
We know from Bianco et al. \cite{BiancoMM10,BiancoMM12} that
the Graded Computation Tree Logic (GCTL) 
is at least exponentially more succinct than the graded $\mu$-calculus. 
We then describe  a linear embedding of
GCTL into $\mu$TLIN. 
A precise definition of GCTL formulas is  first given.
\begin{defi}[Syntax]
The set of Graded Computation Tree Logic formulas is inductively defined
by the following grammar.
\[
\f:= \p \mid \neg \f \mid \f \vee \f \mid E^{>k}X \f \mid E^{>k}G \f \mid E^{>k} \f U \f
\]
\end{defi}

Formulas are also interpreted as node subsets of  finite tree structures.
Proposition are also used as node labels, 
and the boolean operators are interpreted as expected.
Formula $E^{>k} X \f$ is true in nodes with more than $k$ children where $\f$ holds.
$E^{>k} G \f$ holds in nodes with  more than $k$ downward paths
leading to a leaf, such that $\f$ is true in each path node.
And formula $E^{>k} \f U \ff$ holds in nodes $n_0$ with more than
$k$ downward paths $n_0,\ldots,n_k$, such that $\ff$ holds in $n_k$
and $\f$ is true in $n_i$ for every $i<k$.

The {\em all but graded operator} $A^{\leq k}$  is defined as follows:
\begin{align*}
& A^{\leq k}X\f \equiv \neg E^{> k}X \neg \f,
&& A^{\leq k}G \f \equiv \neg E^{>k}F \neg \f,\\
& E^{>k}F \f \equiv E^{>k} \top U \f,
&& A^{\leq k}\f U \ff \equiv \bigvee_{k_1+k_2=k}
\neg\left( E^{>k_1}\left[ \neg \ff U (\neg \f \wedge \neg\ff)\right] \vee E^{>k_2}G \neg \ff\right).
\end{align*}
$A^{\leq k}X \f$ selects nodes with  at most $k$ children where $\f$ does not hold;
$A^{\leq k}G \f$ restricts to at most $k$ the number of downward paths leading to a leaf,
 such that $\f$ does  not hold in each path node;
$E^{>k} F \f$ counts at least $k$ paths where $\f$ holds at least once; and
$A^{\leq k} \f U \ff$ constrains to at most $k$ the number of downward paths
such that the following does not hold: 
$\f$ and $\neg \ff$ are true and in each path node except the last one
where $\ff$ is true.

In order to give a precise GCTL semantics we first describe some useful notations about
downward paths.
\begin{defi}[Children path] \label{def:childpath}
Given a tree structure $T$, 
a children path $\alpha^{\nod_k}_{\nod_0}$ starting at node $\nod_0$ and ending at node $\nod_k$ 
 is a finite set of nodes  $\left\{ n_0,n_1,\ldots,n_k\right\}$, such that
$n_{i+1}\in \Rel(n_i,\fc)$ for $i=0,\ldots,k$.
If the ending node $n_k$ is a leaf, that is, $\Rel(\nod_k,\fc)=\emptyset$,
then we may avoid to write the ending node  $\alpha_{\nod_0}$.
If the starting and the ending node is the same, then the path is defined as the singleton 
$\alpha^n_n=\{n\}$.
\end{defi}

\begin{defi}[Semantics]
Given a tree structure $T$, the interpretation of GCTL formulas is given as follows.
\begin{align*}
 \sem{\p}{T}{}=& \left\{\nod \in \Lab(\p)\right\} \\
 \sem{\neg \f}{T}{} =&  \Nodes \setminus \sem{\f}{T}{}\\
 \sem{\f \vee \ff}{T}{} =&  \sem{\f}{T}{V} \cup \sem{\ff}{T}{V} \\
 \sem{E^{>k}X \f}{T}{} = & \left\{\nod \mid |\Rel(\nod,\fc) \cap \sem{\f}{T}{} |>k\right\}\\
 \sem{E^{>k}G \f}{T}{}=&\left\{\nod \mid 
\left|\left\{ \cpath_\nod  \mid \cpath_\nod \subseteq \sem{\f}{T}{}\right\}\right|>k \right\} \\
 \sem{E^{>k} \f U \ff}{T}{}\ =&  
\left\{\nod_0 \mid  \left| \left\{ \alpha_{\nod_0}^{\nod_k}\neq\emptyset \mid
  \nod_k\in \sem{\ff}{T}{}, \alpha_{\nod_0}^{\nod_{k-1}}\subseteq\sem{\f}{T}{} \right\}\right|>k \right\}
\end{align*}
\end{defi}

As expected, GCTL formulas can be described in terms of $\mu$TLIN formulas.
We now give a precise definition of this embedding.
\begin{defi}[GCTL embedding]
The function $F$ from GCTL formulas to $\mu$TLIN formulas is defined as follows:
\begin{align*}
& F(\p) = \p
&& F(\neg \f) = \neg F(\f) \\
& F(\f \vee \ff)= F(\f) \vee F(\ff)
&& F(E^{>k}X \f)= \nom \wedge  \cf{\left(F(\f)\wedge 
 \mf{\ifc}\left[\nom\wedge \nom=1\right] \right)}>k
\end{align*}
\vspace{-0.6cm}
\begin{align*}
 F(E^{>k}G \f)=& \nom \wedge \cf{\left( \neg\mf{\fc}\top \wedge 
		\muf{\x}F(\f) \wedge \left[ \mf{\ifc}\x \vee \left(\nom\wedge \nom=1\right) \right] \right)}>k \\
 F(E^{>k}\f U \ff)=& \nom \wedge \cf{\left(\ff \wedge \left[ \nom \vee 
	\mf{\ifc}\muf{\x}F(\f) \wedge \left(\mf{\ifc}\x \vee \left[\nom\wedge \nom=1\right]\right)\right]\right)}>k
\end{align*}
\end{defi}

\begin{thm}[Embedding]
For any $GCTL$ formula $\f$,  tree $T$ and  valuation $V$, we have that:
$$ \sem{\f}{T}{}\neq\emptyset \text{ if and only if } \sem{F(\f)}{T}{V}\neq \emptyset$$
and $F(\f)$ has linear size with respect to $\f$.
\end{thm}
\begin{proof}
By induction on the structure of the input formula.

The base case, when the formula is a proposition, is trivial.
 The cases
 of disjunction and negation are immediate by induction.

Consider now the case when the input formula has the form $E^{>k} X \f$.
By induction we know that $\f$ is satisfiable by $T$, if and only if, $F(\f)$ also does.
Now, it is easy to see that $F(\f) \wedge \mf{\ifc}\top$ selects all
the children nodes where $\f$ is true.
Then $\cf{\left( F(\f) \wedge \mf{\ifc}\left[\nom \wedge \nom=1\right]\right)}>k$
is true when the single node marked by $\nom$ has more than $k$
children where $\f$ holds.
Therefore $F(E^{>k}X \f)$ is satisfiable by $T$, if and only if, $E^{>k}X \f$ also does.

Consider now the case for $E^{>k}G\f$. 
By induction we know that $T$ satisfies $\f$, if and only if, $T$ also satisfies $F(\f)$.
Now, recall that $E^{>k}G\f$ is actually counting children paths where $\f$ is true in
each node of the paths. 
Since each node can have one parent only, then each path in $T$ can be distinguished
by the leaf nodes. We can count leaf nodes, and hence paths, with  formula $(\neg\mf{\fc}\top)>k$.
Paths starting at a node $\nom$ where $\f$ is true at each node can be denoted by
$\muf{\x}F(\f) \wedge \left[ \mf{\ifc}\x \vee \left(\nom\wedge \nom=1\right)\right]$.
It is now easy to see that $T$ satisfies
$F(E^{>k}G\f)$,
if and only if, there are at least $k$ children paths starting at node $\nom$, such that
$\f$ holds at each node of the paths.

The remaining case is analogous.

Regarding the size of translation, it is clear that $F$ does not introduce duplications, and
the proof also goes straightforward by induction on the structure of the input formula.
\end{proof}

In order to show that $\mu$TLIN is at least exponentially more succinct than
the graded $\mu$-calculus, we then first define the logic.
\begin{defi}[Graded $\mu$-calculus]
The set of formulas of the graded $\mu$-calculus is defined by the following grammar.
\begin{align*}
\f:= \p \mid \x \mid \neg \f \mid \f \vee \f \mid \mf{\m}\f \mid \muf{\x}\f 
\mid E^{>k}X \f
\end{align*}
\end{defi}
Modalities does not include two-way navigation, that is, $\m\in\{\fc,\fs\}$.
Formulas are interpreted as node subsets of a given tree structure.
The interpretation in the formula fragment corresponding to $\mu$TLIN is
the same as in $\mu$TLIN.
The formula $E^{>k}X \f$ is interpreted as in GCTL.

We now recall a Theorem from Bianco et al.  regarding the exponential
succinctness of GCTL with respect to the graded $\mu$-calculus.
\begin{thm}[GCTL succinctness \cite{BiancoMM10,BiancoMM12}] \label{theo:succ}
There is a GCTL formula $\f$, such that every equivalent graded $\mu$-calculus formula
has exponential size with respect to $\f$.
\end{thm}

From Theorems~\ref{sec:succincteness} and \ref{theo:succ}, it is then easy to infer an exponential succinctness
of $\mu$TLIN formulas with respect to the graded $\mu$-calculus.
\begin{cor}[$\mu$TLIM succinctness]
For any tree $T$ and valuation $V$, 
there is a $\mu$TLIN formula $\f$, such that  every graded $\mu$-calculus
formula $\ff$ is that  if 
$$\sem{\f}{T}{V}\neq\emptyset \text{ if and only if }\sem{\ff}{T}{V}\neq\emptyset,$$
then $\ff$ has exponential size with respect to $\f$.
\end{cor}

\section{Decidability} \label{sec:decidability}
In this Section, we show that
the $\mu$TLIN is decidable.
This is achieved by a reduction to the two-way $\mu$-calculus \cite{DBLP:conf/icalp/Vardi98}.
Before describing the reduction, we first need to recall a well-known bijection between
binary and $n$-ary trees.

\subsection{Binary trees}
 There is well-known
bijection between $n$-ary unranked trees and binary unranked trees \cite{HosoyaVP05}.
One of the edges in the binary trees represents the {\em first child relation},
whereas the other edge represent the {\em following sibling relation}.
In Figure~\ref{fig:bijection} there is a graphical representation of this bijection.
Therefore, from now on, without loss of generality, we will consider binary trees only.

\begin{figure}[t]
\begin{tikzpicture}[level/.style={sibling distance=20mm},scale=0.75]
\node [circle,draw,fill,color=black] (n0) {$n$} 
	child {node [circle,draw,fill,color=gray] (n1) {$n$} }
	child {node [circle,draw,fill,color=gray] (n2) {$n$}
		child {node [circle,draw,fill,color=lightgray] (n21) {$n$}}
		child {node [circle] (n20) {$\ldots$} }
		child {node [circle,draw,fill,color=lightgray] (n2k) {$n$}}
	}
	child {node [circle,draw,fill,color=gray] (n3) {$n$} }
	child {node [circle] (n31) {$\ldots$} }
	child {node [circle,draw,fill,color=gray] (nk) {$n$} }
        child [grow=right,level distance = 7cm] {node [circle,draw,fill,color=black] (m0) {n} 
                edge from parent[draw=none]
		child [grow=down,level distance = 15mm]{node [circle,draw,fill,color=gray] (m1) {$n$}
			child [grow=right] {node [circle,draw,fill,color=gray] (m2) {$n$} 
			        child [grow=right] {node [circle,draw,fill,color=gray] (m3) {$n$} 
					child {node [circle] (m31) {$\ldots$}
						child {node [circle,draw,fill,color=gray] (mk) {$n$}}
					}
				}
				child [grow=down] {node [circle,draw,fill,color=lightgray] (m21) {$n$} 
					child [grow=right] {node [circle] (m20) {$\ldots$}
						child {node [circle,draw,fill,color=lightgray] (mk) {$n$}}
					}
				}
			}
		 }
        }
;
\end{tikzpicture}

\caption{Example of the bijection between $n$-ary and binary trees.}
\label{fig:bijection}
\end{figure}
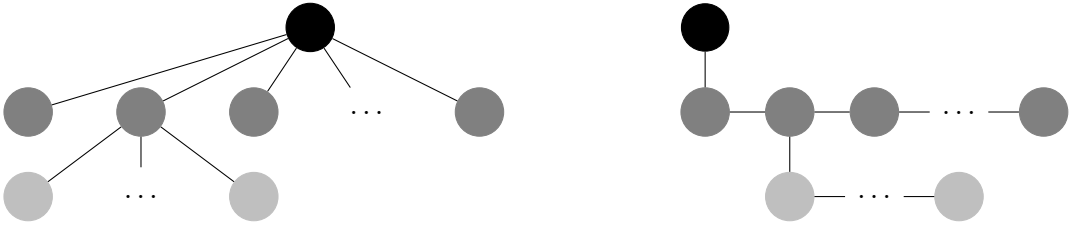

At the logic level, we now reinterpret the modal formula $\mf{\m}\phi$ as follows:
\begin{itemize}
\item formula $\mf{\fc}\phi$ selects the nodes where $\phi$ holds in its first child;
\item formula $\mf{\ifc}\phi$ selects the nodes whose parent satisfy $\phi$;
\item $\mf{\fs}\phi$ holds in nodes where $\phi$ is satisfied by its following sibling; and
\item $\mf{\ifs}\phi$ satisfies nodes such that $\phi$ holds in its previous sibling.
\end{itemize}

\begin{prop}
Consider a bijection $f$ from $n$-ary trees to binary trees, as the one in  \cite{HosoyaVP05}. We have the following:
\begin{itemize}
\item for any $n$-ary tree $T$, valuation $V$, and $\mu$TLIN formula $\f$, 
there is  a $\mu$TLIN formula $\ff$ such that
$$\sem{\f}{T}{V}=\sem{\ff}{f(T)}{V};$$
\item and for any binary tree $B$, valuation $V$, and $\mu$TLIN formula $\ff$,
there is a $\mu$TLIN formula $\f$ such that
$$\sem{\ff}{B}{V}=\sem{\f}{f^{-1}(B)}{V}.$$
\end{itemize}
\end{prop}
\begin{proof}
Consider the first item.
We proceed by induction on the structure of $\f$.
The base and most inductive cases are immediate. We consider the modal case only.
If the input formula has the form $\mf{\fc}{\fff}$, then $\ff$ is
$\mf{\fc}\muf{x}\fff^\prime \vee \mf{\fs} \x$, 
where $\fff^\prime$ is the equivalence (by induction) of $\fff$.
When the input formula is $\mf{\ifc}{\fff}$, then $\ff$ is 
$\muf{\x}\mf{\ifc}\fff^\prime \vee \mf{\ifs}\x$.
The cases for $\mf{\fs}\fff$ and $\mf{\ifs}\fff$ are analogous.
The second item  is trivial: $\ff$ is defined as $\f$.
\end{proof}


\subsection{Reduction}
We now provide a reduction from $\mu$TLIN to the two-way $\mu$-calculus,
that is, we will describe an encoding of counting formulas $\cf{\f}{>k}$
into plain $\mu$-calculus formulas.
For this purpose, we first define a $\mu$-calculus formula counting from the root.

\begin{defi}\label{def:enconding}
 We define the following formulas for $i>1$:
\begin{align*}\
 C_0^\f =&\muf{\x}{\f \vee \mf{\fc}{\x} \vee \mf{\fs}{\x}}\\
C_1^\f= & \muf{x} \left(\f \wedge \left(\mf{\fc}{C^\f_{0}} \vee \mf{\fs}{C^\f_{0}}\right)\right)\vee
			 \left (  \neg\f \wedge 
                \mf{\fc}{C^\f_{0}} \wedge \mf{\fs}{C^\f_{0}}  \right) \vee
		 \mf{\fc}{\x} \vee \mf{\fs}{\x}\\
 C_i^\f =& \muf{x} \left(\f \wedge  \left(\mf{\fc}{C^\f_{i-1}} \vee \mf{\fs}{C^\f_{i-1}} \vee \bigvee_{k_1+k_2=i-2}  \mf{\fc}{C^\f_{k_1}} \wedge \mf{\fs}{C^\f_{k_2}}\right) \right) \vee \\
    & \left (  \neg\f \wedge 
       \bigvee_{k_1+k_2=i-1}  \mf{\fc}{C^\f_{k_1}} \wedge \mf{\fs}{C^\f_{k_2}}  \right)
		\vee \mf{\fc}{\x} \vee \mf{\fs}{\x}
\end{align*}
\end{defi}
From the root node, $C_k^\f$ counts at least $k+1$ nodes satisfying $\f$. 
In Figure~\ref{fig:Cp} there is an example model for $C^{p_1}_3$ holding at the root.
$C^{p_1}_3$ counts at least $4$ nodes named $p_1$.

\begin{figure}[t]
\begin{center}
\begin{tikzpicture}[level/.style={sibling distance=25mm},level distance=10mm]
\node [circle,draw] (z){$p_1$}
	child {node [circle,draw] (a) {$p_1$} 
		child {node [circle,draw] (a1) {$p_1$}	
			child [grow=up,level distance = 7mm] {node (r) {$C^{p_1}_0$} edge from parent[draw=none]}
		}
		child [grow=up,level distance = 7mm] {node (r) {$C^{p_1}_2$} edge from parent[draw=none]}
		child {node [circle,draw] (a2) {$p_2$}
		         child {node [circle,draw] (a21) {$p_1$}
				child [grow=up,level distance = 7mm] {node (r) {$C^{p_1}_0$} edge  from parent[draw=none]}
			}
			child [grow=up,level distance = 7mm] {node (r) {$C^{p_1}_0$} edge from parent[draw=none]}
			child {node [circle,draw] (a22) {$p_2$}}
		}
	}
	child [grow=up,level distance = 7mm] {node (r) {$C^{p_1}_3$} edge from parent[draw=none]}
	child [sibling distance = 35mm]{node [circle,draw] (v) {$p_2$} 
		child {node [circle,draw] (v2) {$p_2$}}
		child {node [circle,draw] (v2) {$p_2$}}
	}
;
\end{tikzpicture}
\end{center}
\caption{Example model for $C^{p_1}_3$ holding at the root}
\label{fig:Cp}
\end{figure}
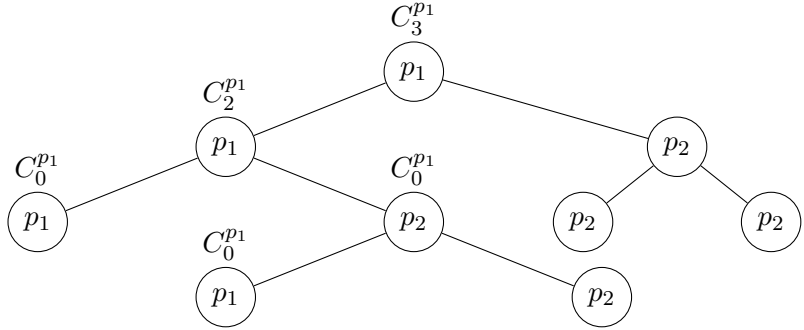

Recall that, in a tree, the root is the only node without a parent, hence
the root $r$ can be denoted by the formula
 $\neg\mf{\ifc}{\top} \wedge\neg\mf{\ifs}{\top}$.
We can thus reach the root  from any other node with the following formula:
$$\muf{\x} r \vee \mf{\ifc}\x \vee \mf{\ifs}\x$$
Now, with the help of $C^\f_k$, we can  now show how to encode counting formulas 
into the simple $\mu$-calculus (without counting constructs).
\begin{lemma}\label{lemma:encoding}
For any tree $T$ and valuation $V$, we have the following:
\begin{align*}
 \sem{\cf{\f}{>k}}{T}{V} =& \sem{\muf{\x}{(C^\f_k \wedge r) \vee \mf{\ifc}{\x} \vee \mf{\ifs}{\x}}}{T}{V}
\end{align*}
\end{lemma}
\begin{proof}
The proof goes by induction on $k$ in $C^\f_k$.
The base cases $C_0^\f$ and $C_1^\f$ are trivial.
For the induction step we distinguish two cases:
\begin{itemize}
\item Assume $\f$ holds at the root, we have then $1$ occurrence of $\f$.
It is then  easy to see by induction that
$$\mf{\fc}{C^\f_{k-1}} \vee \mf{\fs}{C^\f_{k-1}} \vee \bigvee_{k_1+k_2=k-2}  \mf{\fc}{C^\f_{k_1}} \wedge \mf{\fs}{C^\f_{k_2}}$$ 
counts $k$ occurrences of $\f$. There are then $k+1$ occurrence of $\f$.
\item Assume $\f$ does not hold at the root. Then there are two subcases:
\begin{itemize}
\item There are occurrences of $\f$ in both subtrees, in which cases
by induction we know that
$$\bigvee_{k_1+k_2=k-1}  \mf{\fc}{C^\f_{k_1}} \wedge \mf{\fs}{C^\f_{k_2}} $$
counts $k+1$ occurrence of $\f$.
\item The other case is when there are not occurrences of $\f$ in one of the subtrees.
We then apply recursion on the subtrees ($\mf{\fc}\x \vee \mf{\fs}\x$). 
The rest of the proof is immediate by induction on the height of the tree model.\qedhere
\end{itemize}
\end{itemize}
\end{proof}

\noindent Now that we can encode the counting formulas into plain two-way $\mu$-calculus, 
then we can infer that $\mu$TLIN is decidable due to the fact that
$\mu$-calculus is decidable.
However, the encoding of counting formulas results in exponentially larger
$\mu$-calculus formulas.
 
\begin{thm} \label{theo:expressiveness}
$\mu$TLIN  is decidable in double exponential time.
\end{thm}
\begin{proof}
Observe  in Definition~\ref{def:enconding} that $C^\f_k$ encodes numerical constraints by nesting $k$ modalities
on $\f$.
That is, $\mf{\m}{C^\f_{k-1}},\mf{\m}{\mf{\m}{C^\f_{k-2}}},\ldots,\mf{\m}{\ldots\mf{\m}{C^\f_0}}$ are all subformulas of
$C^\f_k$. 
Since $k$ is in binary form,
this implies that there are $2^k$ different occurrences of $\f$ in $C_k^\f$.
That is, the size of $C_k^\f$ is exponentially greater than the sum of the sizes 
of $\f$ and $k$.
Now, by the fact that the $\mu$-calculus is EXPTIME-complete \cite{BonattiLMV06}, and by Lemma~\ref{lemma:encoding},
we conclude the doubly exponential time complexity bound.
\end{proof}

The  graded  $\mu$-calculus \cite{DBLP:conf/cade/KupfermanSV02} was also shown to be
decidable by a reduction to the plain two-way $\mu$-calulculus by B\'arcenas in 
\cite{barcenasthesis}, then
the expressive power of $\mu$TLIN, the graded $\mu$-calculus and
the plain two-way $\mu$-calculus all coincide.

Being $\mu$TLIN decidable, and by Theorems~\ref{theo-paths} and \ref{theo-types},
we can then use as a reasoning framework for XPath queries with
schema and counting constraints. However, the complexity bound 
for decidability can be improved.
In the rest of the paper, we will describe a
satisfiability algorithm with single exponential time complexity.
Before defining the algorithm, we first describe a  Fischer-Ladner representation of tree models.

\section{Fischer-Ladner Trees} \label{sec-trees}
This is a section of preliminaries for the satisfiability algorithm. 
It is described a syntactic representation of tree models.
 
For the algorithm, we consider formulas in negation normal form (NNF) only.
\begin{defi}[Negation Normal Form]
In the negation normal form $\text{nnf}(\f)$ of a formula $\f$,
  negation occurs only immediately above of propositions, $\top$ and modal subformulas $\mf{\m}{\top}$.
 This is obtained by the following rules together with the usual DeMorgan's: 
\begin{align*}
\neg\mf{\m}{\f} =& \mf{\m}{\neg\f} \vee \neg\mf{\m}{\top},  
&\neg(\cf{\f}{> k})=& \cf{\f}{\leq k}, \\
\neg(\cf{\f}{\leq k})=& \cf{\f}{> k},
&\neg\muf{\x}{\f} =& \muf{\x}{\neg\f}\subst{\x}{\neg\x}.
\end{align*}
\end{defi}
Note that, for technical convenience, we  consider an extension of formulas.
 This extension consists of {\em less than} counting formulas $\cf{\f}{\leq k}$
and the {\em true} formula $\top$ with the obvious semantics.

We require some notation before defining the Fischer-Ladner closure.

Since  integers associated to  counting constraints are assumed to be in binary form,
we thus define counter formulas as a boolean combination of propositions denoting an integer number. For example, for a sequence of propositions $\p_1,\p_2,\ldots$, the integer $1$ is written $\p_1 \wedge \bigwedge_{i>1} \neg \p_i$, and the integer $5$ ($101$ in binary) is written $\p_3 \wedge \neg \p_2 \wedge p_1 \wedge \bigwedge_{i>4}\neg\p_i$.
The amount of propositions required to define the counters of  formula $\f$ is  bounded by $\maxk(\f)$.
\begin{defi}
We define $\maxk(\f)$ as follows:
\begin{align*}
 &\maxk(\p)=\maxk(x)=\maxk(\top)=0 \\
 &\maxk(\mf{\m}{\f})=\maxk(\neg \f)=\maxk(\muf{\x}{\f})=\maxk(\f) \\
 &\maxk(\f_1\vee \f_2)=\maxk(\f_1 \wedge \f_2)=\maxk(\f_1)+\maxk(\f_2)\\ 
 &\maxk(\cf{\f}{\# k}) =\maxk(\f)+(k+1)
\end{align*}
When clear from the context, we often simply write $\maxk$.
\end{defi}
Definitions of counters and flags is now given.
\begin{defi}[Counters and flags]
For a counting subformula $\cf{\f}{\#k}$ of a given formula:
\begin{itemize}
\item a counter $\f^{k^\prime}$ set to $k^\prime$ is a sequence of fresh propositions occurring positively in the binary coding of the integer $k^\prime$; and
\item a flag $\f^{\# k}$ is a fresh proposition.
\end{itemize}
\end{defi}
For instance, for the integer $5$ coded as $c_2\wedge\neg c_1 \wedge c_0$, we
write $\f^5$ to denote $c_2,c_0$, where $c_i$ are the corresponding propositions
for the counting formula $\cf{\f}{\# k}$.

The Fischer-Ladner closure of a given formula is the set of its subformulas together with their negation normal form, such that the fixed-points are expanded once. 
Additionally, a counter and a flag for each counting subformula are also considered in the closure.  
All these information is obtained with the help of  the relation $\R$.
\begin{defi}
We  define the following binary relation $\R$ over formulas for $i=1,2$:
\begin{align*}
&\R(\f,\text{nnf}( \f))&\R(\f_1 \wedge \f_2,\f_i) &&\R(\f_1\vee \f_2,\f_i) \\
&\R(\mf{\m}{\f},\f) &\R(\muf{\x}{\f},\f\subst{\muf{\x}{\f}}{\x}) &&\R(\cf{\f}{\# k},\f) \\
&\R(\cf{\f}{\# k}, \f^{\maxk}) &\R(\cf{\f}{\# k},\flag{\f}{\# k}) &&\R(\cf{\f}{\# k},\ff)
\end{align*}
where  $\ff=\muf{\x_1}{(\muf{\x_2}{\f \vee \mf{\fc}{\x_2} \vee \mf{\fs}{\x_2}}) \vee \mf{\ifc}{\x_1} \vee \mf{\ifs}{\x_1}}$.
Notice that if $\f$ is true in a model, then $\ff$ is true in every node of the model.
We use $\psi$ to provide the necessary information for $\f$ to navigate through the entire model.
\end{defi}

We are now ready to define the Fischer-Ladner closure.
\begin{defi}[Fischer-Ladner Closure]
The Fischer-Ladner closure of a given formula $\f$ is  defined as $\flc(\f)=\flc(\f)_k$, such that $k$ is the smallest integer satisfying $\flc(\f)_{k+1}=\flc(\f)_k$, where: 
\begin{align*}
& \flc(\f)_0 =\{\f\} \\
 &\flc(\f)_{i+1} = \flc(\f)_i \cup \{\ff ^\prime\mid \R(\ff,\ff^\prime),\ff \in \flc(\f)_i\}
\end{align*}
\end{defi}

The lean set of a given formula contains  propositions, modal and counting subformulas, together with counters and flags.
\begin{defi}[Lean]
Given a formula $\f$ and a proposition $\p^\prime$ not occurring in $\f$, we  define the lean as follows for all $\m\in \M$:
\begin{align*}
 \lean(\f) =\{\p,\mf{\m}{\ff},\cf{\ff}{\# k},\ff^\maxk,\flag{\ff}{\# k}\in\flc(\f)\}\cup\{\mf{\m}{\top},\p^\prime\}
\end{align*}
\end{defi}
The lean set contains all the required information to define  tree nodes: 
propositions serve as labels, modal subformulas define the topology of the tree, 
and counters and flags serve to verify  counting subformulas.

As in \cite{BarcenasGLS11, CalvaneseGLV10, GenevesLS07}, the single exponential time complexity of the satisfiability algorithm mainly relies
in the size of the lean set (tree nodes are defined as subsets of the lean).
Since counters  are coded in binary, it is then easy to see that the size of the lean set is not significantly increased with respect to the original formula. 
\begin{lemma}\label{lem-lean}
The cardinality of $\lean(\f)$ is linear with respect to the size of $\f$.
\end{lemma}
\begin{proof}
The proof goes by structural induction on $\f$.

We consider only the case for counting subformuals $R^{FL}(\cf{\f^\prime}{\# k},\psi)$.
Now recall that a counter is defined in terms of  a boolean combination of propositions,
that is, for each counting subformula $\cf{\f^\prime}{\# k}$ only $\log(\maxk)$ ($\maxk$ in binary)
new propositions are introduced in the lean. Since the size of $\cf{\f^\prime}{\# k}$ is
defined by $|\f^\prime|+\log(k+1)$, and by the definition of $\maxk$, the counters then produce no increment in the size of the lean. 
\end{proof}

\begin{exa} \label{exa:lean}
Consider the following formulas for  $\m\in\{\fc,\fs,\ifc,\ifs\}$:
\begin{align*}
 \f=& \cf{\left[\left(\cf{\p_1}{>1}\right) \wedge \p_2\right]}{>4}
&\ff=& \left(\cf{\p_1}{>1}\right) \wedge \p_2\\
\f_0=& \muf{\x} \ff \vee \bigvee_{\forall \m} \mf{\m}\x 
&\ff_0 =& \muf{\x} \p_1 \vee \bigvee_\m \mf{\m}\x
\end{align*}
The lean of $\f$ is thus defined as follows for $\m\in\{\fc,\fs,\ifc,\ifs\}$:
\begin{align*}
\lean(\f)=&\{\p_1,\p_2,\f,\cf{\p_1}{>1},\ff^7,\p_1^7,\ff^{>4},\p_1^{>1},
            \mf{\m}\f_0,\mf{\m}\ff_0,\p^\prime,\mf{\m}\top\}
\end{align*}
$\maxk=7$.  Now recall that $\f^7$ denote $3$ propositions that serve to express
the binary coding of the integers from $0$ to $7$.
\end{exa}

We are now ready to define the syntactic notion of tree nodes.
\begin{defi}[$\f$-Nodes]
Given a formula $\f$, a $\f$-node $\nod^\f$ is defined as a subset of $\lean(\f)$, such that:
\begin{itemize}
\item at least one proposition of $\f$ occurs;
\item if $\mf{\m}{\ff}$ occurs, then $\mf{\m}{\top}$ also does;
\item both $\mf{\ifs}{\top}$ and $\mf{\ifc}{\top}$ can not occur;
\item counting formulas are always present;
\item exactly one counter for each counting formula is present, i.e., if $\cf{\f}{\# k}\in\nod^\f$, then $\f^{k^\prime}\in\nod^\f$; 
\item counters must be consistent with counting formulas and flags, i.e.,
$\ff^{k_0},\cf{\ff}{\leq k}\in \nod$, if and only if, $k_0\leq k$, and
 $\ff^{k_0},\flag{\ff}{> k}\in \nod$, if and only if, $k_0 > k$.
\end{itemize}
\end{defi}

The set of $\f$-nodes is written $\Nd^\f$. 
If the context is clear, we often call a $\f$-node simply a node,
and we write $\nod$ instead of $\nod^\f$.

We now define trees as triples $(\nod,X_1,X_2)$, where $\nod$ is the root of the tree and $X_1$ and $X_2$ are the respective left and right subtrees.
\begin{defi}[Fischer-Ladner trees]
 Given a formula, a Fischer-Ladner tree, or simply a tree, is inductively defined as follows:
\begin{itemize}
\item the empty set $\emptyset$ is a tree;
 \item the triple $(\nod^\f,X_1,X_2)$ is also a tree, provided that $X_1$ and $X_2$ are also trees.
\end{itemize}
\end{defi}


\begin{exa}\label{exa:tree}
Consider $\f,\ff,\f_0,\ff_0$ from Example~\ref{exa:lean}.
We define the following syntactic tree model for $\f$:
\begin{align*}
T=&(\nod_0,(\nod_1,(\nod_3,\emptyset,\emptyset),(\nod_4,\emptyset,\emptyset)),
    (\nod_2,(\nod_5,\emptyset,\emptyset),(\nod_6,\emptyset,\emptyset))
\end{align*}
where
\begin{align*}
 \nod_0 =& \{\p_2,\f,\cf{\p_1}{>1},\p_1^2,\p_1^{>1},\ff^5,\ff^{>4}, 
        \mf{\fc}\ff_0,\mf{\fs}\ff_0,\mf{\fc}\f_0,\mf{\fs}\f_0,\mf{\fc}\top,\mf{\fs}\top\} \\
\nod_1 =& \{\p_2,\f,\cf{\p_1}{>1},\p_1^2,\p_1^{>1},\ff^1,\mf{\fc}\ff_0,\mf{\fs}\ff_0,
    \mf{\ifc}\ff_0,\mf{\fc}\f_0,\mf{\fs}\f_0,\mf{\ifc}\f_0,
        \mf{\fc}\top,\mf{\fs}\top,\mf{\ifc}\top\} \\
\nod_2 =& \{\p_2,\f,\cf{\p_1}{>1},\p_1^2,\p_1^{>1},\ff^3,\mf{\fc}\ff_0,\mf{\fs}\ff_0,
    \mf{\ifs}\ff_0,\mf{\fc}\f_0,\mf{\fs}\f_0,\mf{\ifs}\f_0,
\mf{\fc}\top,\mf{\fs}\top,\mf{\ifs}\top\} \\
\nod_3 =& \{\p_1,\f,\cf{\p_1}{>1},\p_1^1,\mf{\ifc}\f_0,\mf{\ifc}\ff_0,\mf{\ifc}\top\}\\
\nod_4 =& \{\p_1,\f,\cf{\p_1}{>1},\p_1^1,\mf{\ifs}\f_0,\mf{\ifs}\ff_0,\mf{\ifs}\top\}\\
\nod_5 =& \{\p_2,\f,\cf{\p_1}{>1},\ff^1,\mf{\ifc}\f_0,\mf{\ifc}\ff_0,\mf{\ifc}\top\}\\
\nod_6 =& \{\p_2,\f,\cf{\p_1}{>1},\ff^1,\mf{\ifs}\f_0,\mf{\ifs}\ff_0,\mf{\ifs}\top\}\\
\end{align*}
Figure~\ref{fig:algo} depicts  a graphical representation of $T$.
\end{exa}

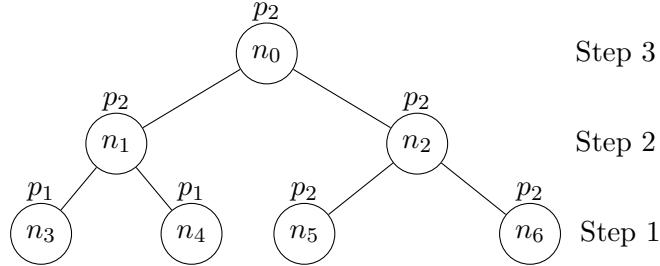
\begin{figure}[t]
\begin{center}
\begin{tikzpicture}[level/.style={sibling distance=25mm/#1}, scale=0.8]
\node [circle,draw] (z){$\nod_0$}
	child {node [circle,draw] (a) {$\nod_1$} 
		child {node [circle,draw] (a1) {$\nod_3$}	
			child [grow=up,level distance = 7mm] {node (r) {$\p_1$} edge from parent[draw=none]}
		}
		child [grow=up,level distance = 7mm] {node (r) {$\p_2$} edge from parent[draw=none]}
		child {node [circle,draw] (a2) {$\nod_4$}
			child [grow=up,level distance = 7mm] {node (r) {$\p_1$} edge from parent[draw=none]}
		}
	}
	child [grow=up,level distance = 7mm] {node (r) {$\p_2$} edge from parent[draw=none]}
	child {node [circle,draw] (v) {$\nod_2$} 
		child {node [circle,draw] (a1) {$\nod_5$}	
			child [grow=up,level distance = 7mm] {node (r) {$\p_2$} edge from parent[draw=none]}
		}
		child [grow=up,level distance = 7mm] {node (r) {$\p_2$} edge from parent[draw=none]}
		child [grow=right,level distance = 33mm] {node (r) {Step $2$} edge from parent[draw=none]
			child [grow=up,level distance = 15mm] {node (r) {Step $3$} edge from parent[draw=none]}
		}
		child {node [circle,draw] (a2) {$\nod_6$}
			child [grow=up,level distance = 7mm] {node (r) {$\p_2$} edge from parent[draw=none]}
			child [grow=right] {node (r) {Step $1$} edge from parent[draw=none]}
		}
	}
;
\end{tikzpicture}
\end{center}
\caption{Fischer-Ladner tree model for $\f= \cf{\left[\left(\cf{\p_1}{>1}\right) \wedge \p_2\right]}{>4}$}
\label{fig:algo}
\end{figure}


\section{Satisfiability} \label{sec:sat}
In this Section, we introduce a satisfiability algorithm for $\mu$TLIN in the style of Fischer-Ladner \cite{BarcenasGLS11,DemriL10}.
Tree nodes are defined from an extension of the classical Fischer-Ladner closure.
The extension consists of counters (boolean combination of fresh propositions encoding integer values in binary) that are used to verify counting formulas.
Tree models are built in a bottom-up manner, that is,
starting from the leaves, parent nodes are consistently added until a witness tree for the formula in question is found.
At each step in this process, 
counters must be consistent with the counters of children nodes 
and the formulas that hold in the current parent node.



\subsection{The algorithm}

The satisfiability algorithm, described in  Algorithm~\ref{fig:satalgo}, builds candidate trees in a bottom-up manner:
iteratively, starting from leaf nodes, we check at each step if the input formula is satisfied by candidate trees,
in case the formula is not satisfied, we consistently add parents to previously built trees.
The algorithm returns $1$ if a satisfying tree is found.
In case a satisfying tree could not be found, and no more candidate trees can be built,
then the algorithm returns $0$.

\begin{exa} \label{exa:algorun}
Consider the formula $\f$ defined in Example~\ref{exa:lean}.
Then the Fischer-Ladner tree defined in Example~\ref{exa:tree} 
is built by the algorithm in $3$ steps.
In the first step, all the leaves are considered, that is,
nodes without children, such that the counters are properly initialized
(Definition~\ref{def:leaves}).
It is then easy to see that $\nod_3,\nod_4,\nod_5,\nod_6$ are all leaves.
Since $\p_1$ is occurring in both, $\nod_3$ and $\nod_4$,
then the counter $\p_1^1$ is also in the same nodes.
Since both $\p_2$ and $\cf{\p_1}{>1}$ are in $\nod_5$ and $\nod_6$,
then $\ff=\p_2\wedge\cf{\p_1}{>1}$ is true in both nodes, 
and consequently $\ff^1$ is also in $\nod_5$ and $\nod_6$.
However, none of the leaves satisfies $\f$, then, in the second step, 
$\nod_1$ is added as parent to both $\nod_3$ and $\nod_4$.
$\nod_2$ is also added as parent to $\nod_5$ and $\nod_6$.
Since $\ff$ is true in $\nod_1$ and $\nod_2$,
 then the counter for $\ff$ is incremented in both nodes.
Resulting that in $\nod_1$ we have $\ff^1$, and in $\nod_2$ we have $\ff^3$.
However, none of the trees built in step 2 satisfies $\f$.
In step $3$, $\nod_0$ is then added as parent of $\nod_1$ and $\nod_2$.
Since $\ff$ holds in $\nod_0$, then we update the counter to $\ff^5$,
and $\f$ is then finally satisfied.
This process is depicted in Figure~\ref{fig:algo}.
\end{exa}

\begin{algorithm}[t]
\begin{algorithmic}
\STATE $Y \gets \Nd^\f$
\STATE $\mathcal{X} \gets Leaves(Y)$
\STATE $\mathcal{X}_0 \gets \emptyset$
\WHILE{$\mathcal{X}\neq \mathcal{X}_0$}
\IF {$\mathcal{X}\Vdash \f$}
\RETURN $1$
\ENDIF
\STATE $\mathcal{X}_0\gets \mathcal{X}$
\STATE $(\mathcal{X},Y) \gets Update(\mathcal{X},Y)$
\ENDWHILE
\RETURN $0$
\end{algorithmic}
\caption{Satisfiability Algorithm} \label{fig:satalgo}
\end{algorithm}

We now provide a precise description of the algorithm components.

If a tree $T$ is a model for a formula $\f$, it is said that $T$ satisfies (entails) $\f$.
We now give a precise definition of this entailment relation.
\begin{defi}
The entailment of a formula by a node is defined by:
\begin{align*}
& \frac{}{\nod\vdash \top} && \frac{\f\in\nod}{\nod\vdash \f} && \frac{\f\not\in\nod}{\nod\vdash\neg\f}
&& \frac{\nod\vdash\f \;\;\;\;\; \nod\vdash \ff}{\nod\vdash \f \wedge \ff}\\
& \frac{\nod\vdash \f}{\nod\vdash \f\vee\ff} && \frac{\nod\vdash \ff }{\nod\vdash \f \vee \ff} 
&& \frac{\nod\vdash \f\subst{\muf{\x}{\f}}{\x}}{\nod\vdash\muf{\x}{\f}}
\end{align*}
The entailment relation is now extended for trees and formulas.
 A formula $\f$ is satisfied by a tree $X$, written $X\Vdash \f$,
if and only if,
\begin{itemize}
\item there is a node $\nod$ in $X$, such that $\nod\vdash \f$;
\item formulas of the forms $\mf{\ifc}{\ff}$ and $\mf{\ifs}{\ff}$ do not occur in the root of $X$; and
\item all the flags are in the root. 
\end{itemize}
A set of trees $\mathcal{X}$ entails a formula $\f$, written $\mathcal{X}\Vdash \f$, 
if and only if, there is a tree $X$ in 
$\mathcal{X}$ s.t. $X\Vdash \f$.

The relation $\not\Vdash$ is defined as expected.
\end{defi}

The sef of leaves contains nodes without children. In the leaves, counters are also properly initialized.
\begin{defi}[Leaves] \label{def:leaves}
Given set of nodes $X$, the set of leaves  is defined as follows:
\begin{align*}
Leaves(X)=&\left \{(\nod,\emptyset,\emptyset) \mid \nod\in X, \mf{\fc}{\f},\mf{\fs}{\f}\not\in \nod,
  \left[(\f^1\in \nod, \nod\vdash \f) \text{ or } (\f^0 \in \nod, \nod\not\vdash \f)\right]
\right \}
\end{align*}
\end{defi}

Recall that counting formulas are true in the entire model when satified, then
counting formulas are always present in every $\f$-node.  
The corresponding counters will be updated each time they find a witness. 
Notice that counting subformulas with the form $\cf{\ff}{> k}$ may not be true at earlier steps of the algorithm.
We then use  flags  to identify when those formulas become true, that is, when we find more than $k$ witnesses of $\ff$, we then turn on the flag $\flag{\ff}{> k}$.
Once a flag is turned on, it is copied to  parents  at  each further step. It is then required to have all the flags in the root in order to ensure that counting subformulas $\cf{\ff}{> k}$ are all satisfied.

For the step case in the algorithm, if  newly  built trees do not satisfy the formula, then new candidate trees are constructed by adding a parent to previously built trees. 
This is done by the $Update$ function, which is defined with the help of the following auxiliary functions.

A node $\nod$ containing a modal formula $\mf{\m}{\ff}$ can be linked to another node $\nod^\prime$ through a modality $\m$, if and only if, there is a witness of $\ff$ in $\nod^\prime$, that is, $\nod^\prime\vdash \ff$. 
This notion is defined by the relation $\Delta_\m$.
\begin{defi} \label{def:modalcon}
Given two nodes $\nod_1$, $\nod_2$ and formula $\f$, we say that 
the nodes are modally consistent with respect to the formula $\Delta_\m (\nod_1,\nod_2)$ for $\m\in\{\fc,\fs\}$, if and only if, for all formulas $\mf{\m}{\ff_1},\mf{\dual{\m}}{\ff_2}\in\lean(\f)$,
we have that:
\begin{itemize}
\item $\mf{\m}{\ff_1}\in \nod_1$ if and only if $\nod_2\vdash \ff_1$, and
\item $\mf{\dual{\m}}{\ff_2}\in \nod_2$ if and only if $\nod_1\vdash \ff_2$.
\end{itemize}
\end{defi}
\begin{exa}
Consider the algorithm execution described in Example~\ref{exa:algorun}.
In the second step, when linking $\nod_1$ with $\nod_3$ and $\nod_4$, note that
$\Delta_\fc(\nod_1,\nod_3)$ and $\Delta_\fs(\nod_1,\nod_4)$. 
This is because  $\f_0$ and $\ff_0$ are both true in $\nod_3$ and $\nod_4$,
that is, $\nod_3\vdash \f_0$, $\nod_3\vdash \ff_0$, $\nod_4\vdash \f_0$,
and $\nod_4\vdash \ff_0$.
\end{exa}

When adding parents, it is also necessary to ensure that counting formulas are satisfied.
Recall that, according to the definition of $\f$-nodes, counting formulas and flags are consistent with  counters.
It is then only required to update the counters and to copy the flags that are already in the subtrees.
We have two cases. The first one is when we add a parent to both, a left and a right subtrees. The second case is  when a parent is added to one subtree only.
Consider the first case.
\begin{defi} \label{def:countercon}
 It is said that three nodes $\nod_0,\nod_1,\nod_2$   are consistent with respect to their counters, denoted by $\#(\nod_0,\nod_1,\nod_2)$, 
if and only if,
\begin{itemize}
\item $\ff^{k_0}\in\nod_0$ and $\nod_0\vdash\ff$, if and only if, $\ff^{k_1}\in\nod_1$, $\ff^{k_2}\in\nod_2$ and $k_0=k_1+k_2+1$ if $k_0\leq \maxk$, otherwise $k_0=\maxk$; 
\item $\ff^{k_0}\in\nod_0$ and $\nod_0\not\vdash\ff$, if and only if,  $\ff^{k_1}\in\nod_1$, $\ff^{k_2}\in\nod_2$ and $k_0=k_1+k_2$ if $k_0\leq \maxk$, otherwise $k_0=\maxk$; and
\item if $\flag{\ff}{>k}\in \nod_i$ for any $i\in\{1,2\}$, then $\flag{\ff}{>k}\in \nod_0$.
\end{itemize} 
\end{defi}
The second case ($\#(\nod_0,\nod_i)$) is defined in an analogous manner.
\begin{exa}
Consider again the execution described in Example~\ref{exa:algorun}.
Since $\ff^1\in \nod_1$, $\ff^3\in \nod_2$ and $\nod_0\vdash \ff$,
it is then consistent that $\ff^5\in \nod_0$, and hence $\#(\nod_0,\nod_1,\nod_2)$.
\end{exa}

Recall that the $Update$ function is used to consistently add parents to previously built trees.
Now, with the notions of modal and counter consistency (Definitions~\ref{def:modalcon} and \ref{def:countercon}) already defined,
we are now ready to give a precise description of the $Update$ function.
\begin{defi}
Given a set of trees $\mathcal{X}$ and a set of nodes $Y$, 
the function $Update(\mathcal{X},Y)$ is defined as the tuple $(\mathcal{X}^\prime,Y^\prime)$, such that:
\begin{itemize}
\item $\mathcal{X}^\prime=\{(\nod,X_\fc,X_\fs) \mid \nod\in Y,X_i\in\mathcal{X},\Delta_i(\nod,\nod_i),\#(\nod,\nod_1,\nod_2)\}$,
where $i=\fc,\fs$ and $\nod_i$ is the root of $X_i$; or 
\item $\mathcal{X}^\prime=\{(\nod,X_\fc,X_\fs) \mid \nod\in Y,X_i\in\mathcal{X},\Delta_i(\nod,\nod_i),\#(\nod,\nod_i)\}$ in
case  $X_j=\emptyset$ with $i\neq j$; and 
\item $Y^\prime=Y\setminus\{\nod\}$.
\end{itemize} 
\end{defi}

We now prove that the algorithm is correct. 
We also describe a single exponential bound in the time complexity of the algorithm.

\subsection{Correctness and Complexity} \label{sec-corre}

It is easy to see that the algorithm has a finite number of steps if we notice that the number of nodes is finite and that the $Update$ function is monotone.

In order to show that the algorithm is correct, we then prove it to be sound and complete.

\begin{thm}[Soundness]
If the algorithm returns $1$ for the input formula $\f$, then there is tree model satisfying $\f$.
\end{thm}
\begin{proof}
By assumption,  there is a triple $X$ such that $X\Vdash \f$.
We will now construct a tree model $T$ from $X$.
\begin{itemize}
\item The set of propositions $\Prop$ are the ones in  $\lean(\f)$.
\item The nodes of $T$ are $\Nodes^\f$.
\item We now define the edges of $T$.
For every triple $(\nod,X_1,X_2)$ of $X$, we define $\Rel(\nod,\fc)=\nod_1$ and   $\Rel(\nod,\fs)=\nod_2$, provided that
$\nod_1$ and $\nod_2$ are the respective roots of $X_1$ and $X_2$.
\item We label the nodes in the obvious manner: if $\p\in\nod$, then $\p\in\Lab(\nod)$.
\end{itemize}
It is now shown by structural induction on $\f$ that $T$ satisfies $\f$.
All cases are straightforward. For the case of fixed-point subformulas, recall that there is an equivalent finite unfolding, that is:	$\muf{\x}{\ff} \equiv \f\subst{\muf{\x}{\f}}{\x}$ \cite{BonattiLMV06,BarcenasGLS11}.
\end{proof}

For completeness it is assumed that there is a satisfying tree $T$ for the formula $\f$, and then it is shown that the algorithm returns $1$. The proof comes in two steps: we first construct an equivalent lean labeled version of $T$, and then we show that the algorithm can actually construct such lean labeled  tree.

\begin{defi}
Given a satisfying tree $T$ of a formula $\f$, we define its lean version $X^T$ as follows:
\begin{itemize}
\item $X^T$ has the same nodes and shape than $T$;
\item each node $\nod$ in $X^T$ is labeled with the formulas $\ff$ in $\lean(\f)$ such that
\begin{itemize}
\item $\nod$ in $T$ satisfies $\ff$, and
\item the labels corresponding to the counters are pinned up in a similar manner as the algorithm does, that is,
in an increasing order (with bound $\maxk$) from bottom-up in the tree.  
\end{itemize}
\end{itemize}
\end{defi}

\begin{lemma}
If a tree $T$ satisfies a formula $\f$, then $\f$ is entailed by $X^T$.
\end{lemma}
\begin{proof}
We proceed by induction on the derivation of $\nod\vdash\f$.
Most cases are immediate by induction and the construction of $X^T$.

For the fixpoint case $\muf{\x}{\ff}$, we test $\ff\subst{\muf{\x}{\ff}}{\x}$.
We then proceed by structural induction again. This is also straightforward since 
variables, and hence unfolded fixed-points, can only occur in the scope of a modality or a counting formula.
\end{proof}

One crucial point in the completeness proof is to show that $N^\f$ contains enough nodes to satisfy $\f$.
It is well-known that the standard Fischer-Ladner construction of models provides the required amount of nodes for simple $\mu$-calculus formulas without counting \cite{BonattiLMV06}. 
Since counting subformulas impose bounds on the number of certain nodes, 
it may be required to duplicate $\f$-nodes.
Counters are then introduced in the Fischer-Ladner construction in order to distinguish  potentially identical nodes.
 We now show that counters are introduced in a consistent manner.

\begin{lemma}\label{perro}
Given a satisfying tree $T$ of a formula $\f$, there is a tree entailing $\f$, such that for every path from its root to a leave, there  are not identical $\f$-nodes.
\end{lemma}
\begin{proof}
If every path in $X^T$ does not contain identical nodes, then we are done.

Consider now the case when we have two identical nodes $\nod_1$ and $\nod_2$ in a path of $X^T$.
Without loss of generality, we assume that $\nod_1$ is above $\nod_2$. 
We then proceed to build a tree $X$ from $X^T$, such that $\nod_2$ is grafted upon $\nod_1$.
That is, the path between $\nod_1$ and $\nod_2$ is removed, not including $\nod_1$ but including $\nod_2$.
$\nod_1$ is then linked to the subtrees of $\nod_2$.
$X$ can then be seen as the pruned version of $X^T$.

We now show that $X$ also entails $\f$ by induction on the derivation of $X\vdash \f$.
Most cases are immediate by the construction of $X$ and by induction.


Consider now the case of counting subformulas.
Since these subformulas are true in every node, then the only important thing is to be sure that the counted nodes are not part of the pruned path. This is not possible since the counters in $\nod_2$ are the same than the ones in $\nod_1$, that is,
the counters are not increased between $\nod_1$ and $\nod_2$.
\end{proof}


\begin{thm}[Completeness]
If a formula $\f$ is satisfiable, then the algorithm returns $1$.
\end{thm}
\begin{proof}
By assumption, there is a (Kripke) tree $T$ satisfying $\f$.
By Lemma~\ref{perro}, we know there is a Fischer-Ladner tree $X^T$, obtained from $T$, entailing $\f$,  and whose nodes are all in $\Nd^\f$.
In order to show that $X^T$ is produced by the algorithm,
we now proceed by induction on the height of $X^T$.

The base case is immediate. 

For the induction step, we know that the right and left subtrees of $X^T$, say $X_\fc$ and
$X_\fs$, are already produced by the algorithm, that is, $X_\fc,X_\fs\in \mathcal{X}$.
In order to show that $Update(\mathcal{X},Y)=(\mathcal{X}^\prime,Y^\prime)$,
such that $X^T \in\mathcal{X}^\prime$, please note
 that $\Delta_\fc(n,X_\fc)$ and $\Delta(n,X_\fs)$, where $n$ is the root of $X$.
The fact that $n\in Y$ comes from the consistency of $\maxk$ with respect to
satisfaction of $\f$, which is easily proved by an immediate induction
on the structure of $\f$.
\end{proof}
 
As in \cite{BarcenasGLS11, CalvaneseGLV10, GenevesLS07}, the time complexity of the satisfiability algorithm is single exponential on the
number of nodes (automaton states) introduced by the Fischer-Ladner construction.

\begin{thm}[Complexity]
$\mu$TLIN satisfiability is  EXPTIME-complete.
\end{thm}
\begin{proof}
By Lemma~\ref{lem-lean}, the size of the lean is at most polynomial 
with respect to the formula size.
We then show that the complexity of the algorithm is at most exponential 
with respect to the lean size.

First notice that the size of $\Nd^\f$ is exponentially bounded by the lean size.
Then, in the loop there is  at most an exponential number of steps.

Computing the set $Leaves$ takes exponential time since $\Nd^\f$ is traversed once.

Now note  that testing  the relation $\vdash$ costs linear time with respect to the size of the node. Then the entailments  $\Vdash$ and $\not\Vdash$ take  at the most exponential time .

The $Update$ function costs at the most exponential time by the following facts: traversals on $\mathcal{X}$ and $Y$ take exponential time; and the costs of the relations $\Delta$ and $\#$ are linear.
Since each step in the loop  takes at the most exponential time, we conclude that the overall complexity is single exponential.

Finally, since $\mu$TLIN can encode all finite tree automata and is closed under negation, satisfiability is hard for EXPTIME,
and hence complete.
\end{proof}

 Recall that regular path queries  (XPath) and regular tree expressions (XML schemas), 
extended with counting contructs, can be encoded in terms of the logical formulas with linear size with respect to the original queries and types (Theorems~\ref{theo-paths} and \ref{theo-types}).
We can then conclude that the logic can  be used as an optimal query reasoning framework 
for XML trees.

\begin{cor}
The emptiness, containment and equivalence of CPath queries and CTypes  are decidable in EXPTIME.
\end{cor}

\section{Conclusions} \label{sec-conclu}
We introduced a modal tree logic with counting and  multi-directional navigation.
We also showed that the logic can linearly characterize counting extensions of regular path queries (XPath) and 
regular tree types (XML schemas).
The logic was also shown to be satisfiable in single exponential time even if the numerical constraints are coded in binary.
In consequence, the logic serves as reasoning framework for XML queries and schemas extended with counting constructs.
These constructs restrict the number of multi-directional regular paths.
Since the logic is closed under negation, we can then decide in EXPTIME  typical  reasoning problems 
such as emptiness, containment, and equivalence of XML queries and schemas.
We are currently working on the implementation of the satisfiability algorithm described
in the present work with the use of  Binary Decision Diagrams (BDD's), as previously 
described in \cite{GenevesLS07,DBLP:conf/aiml/TanabeTH08}.

Proving correctness of programs is a crucial part in the verification of software, such as operating or real-time systems. The implementation of efficient high level program structures are often based on balanced tree structures, such as AVL trees, red-black trees, splay trees, etc.
Reasoning frameworks with in-depth counting constraints, such as the ones described in this work, play a major role in the verification of balanced tree structures, as already described in Habermehl et al. \cite{HabermehlIV10} and Manna et al. \cite{MannaSZ07}.
Therefore, we believe it is possible to study the field of applications 
of the reasoning frameworks developed in this work  in the context of 
the verification of balanced tree structures.
Also in the formal verification side, 
the behavior of reactive systems has been extensively studied by means of 
the model checking problem for the $\mu$-calculus \cite{DBLP:conf/fossacs/FerranteM07,CalvaneseGLV10}.
We also consider the model checking problem for $\mu$TLIN as 
as a further research direction.

\paragraph{\bf Acknowledgments.}
This work benefited from the support of Pierre Genev\`es, Nabil Laya\"ida, Denis Lugiez and Alan Schmitt.

\bibliographystyle{alpha}
\bibliography{ref}

\end{document}